\documentclass{arxiv}
\usepackage{amsmath}
\usepackage{amssymb}
\usepackage{amsfonts}
\usepackage{amstext}
\usepackage{amsopn}
\usepackage{amsxtra}
\usepackage{pst-all}
\usepackage[latin1]{inputenc}
\usepackage{dsfont}
\usepackage{mathrsfs}
\usepackage{multirow}
\usepackage{graphicx}
\newcommand\R{{\ensuremath {\mathbb R} }}
\newcommand\C{{\ensuremath {\mathbb C} }}

\newcommand\1{{\ensuremath {\mathds 1} }}
\renewcommand\phi{\varphi}
\newcommand{\alp}{\boldsymbol{\alpha}}

\renewcommand{\to}{\rightarrow}

\newcommand{\cE}{\mathcal{E}}

\newcommand\ii{{\ensuremath {\infty}}}
\newcommand\pscal[1]{{\ensuremath{\left\langle #1 \right\rangle}}}
\newcommand{\norm}[1]{ \left| \! \left| #1 \right| \! \right| }

\renewcommand{\epsilon}{\varepsilon}
\newcommand{\sP}{{\,^1\!P}}
\newcommand{\tP}{{\,^3\!P}}
\newcommand{\tD}{{\,^3\!D}}

\begin{document}
\date{\today}
\title{\Large Symmetry Breaking of Relativistic Multiconfiguration\\  \medskip Methods in the Nonrelativistic Limit}
\authormark{Maria J. ESTEBAN, Mathieu LEWIN \& Andreas SAVIN}
\runningtitle{Failure of Multiconfiguration Methods}

\author{Maria J. ESTEBAN}
\address{CNRS \& Ceremade (CNRS UMR 7534), Universit{\'e} Paris-Dauphine\\ Place
du Mar{\'e}chal de Lattre de Tassigny, 75775 Paris Cedex 16 - France.\\ Email: \email{esteban@ceremade.dauphine.fr}}

\author{Mathieu LEWIN}
\address{CNRS \& Laboratoire de Mathématiques (CNRS UMR 8088), Universit{\'e} de Cergy-Pontoise\\ 2, avenue Adolphe Chauvin, 95 302 Cergy-Pontoise Cedex - France.\\ Email: \email{Mathieu.Lewin@math.cnrs.fr}}

\author{Andreas SAVIN}
\address{CNRS \& Laboratoire de Chimie Théorique (CNRS UMR 7616), Université Pierre et Marie Curie\\ 4 place Jussieu, Case courrier 137, 75252 Paris Cedex 05 - France.\\ Email: \email{andreas.savin@lct.jussieu.fr}}

\maketitle

\bigskip

\begin{abstract}
The \emph{multiconfiguration Dirac-Fock} method allows to calculate the state of relativistic electrons in atoms or molecules. This method has been known for a long time to  provide certain wrong predictions in the nonrelativistic limit. We study in full mathematical details the nonlinear model obtained in the nonrelativistic limit for $Be$-like atoms. We show that the method with $sp+pd$ configurations in the $J=1$ sector leads to a symmetry breaking phenomenon in the sense that the ground state is never an eigenvector of $L^2$ or $S^2$. We thereby complement and clarify some previous studies.

\smallskip

\noindent{\scriptsize\copyright~2009 by the authors. This paper may be reproduced, in its entirety, for non-commercial~purposes.}
\end{abstract}

\tableofcontents

\bigskip\bigskip

\section*{Introduction}
\addcontentsline{toc}{section}{Introduction}
Simulations of relativistic systems of Atomic and Molecular Physics and Chemistry are now of widespread use and the need for reliable methods is stronger than ever \cite{PyySch-04,IndLinDes-05}. The difficulties of doing relativistic simulations are however probably largely underestimated. One the the most common method is the so-called \emph{multiconfiguration Dirac-Fock} (or \emph{Dirac-Hartree-Fock}) theory. This method has been known for a long time to  provide certain wrong predictions in the nonrelativistic limit \cite{WooPyp-80,Pyper-83,HuaKimCheDes-82,KimParMarIndDes-98,FroeseFischer-00,IndLinDes-05}. The purpose of the present article is to explain in full mathematical details the encountered difficulties.

In atomic relativistic calculations, one usually imposes the total angular momentum $J^2=(L+S)^2$ whereas in nonrelativistic calculations both $L^2$ and $S^2$ are imposed. It has been observed by Kim et al in \cite{KimParMarIndDes-98} that a certain multiconfiguration ground state of the symmetry space $J=1$ for $Be$-like (Beryllium-like) atoms, was converging in the nonrelativistic limit to a state which was not an eigenfunction of $S^2$ and $L^2$. This led to erroneous values of certain transition probabilities like spin-forbidden ones.

In the multiconfiguration methods, the wavefunction is taken to be a linear combination of certain \emph{configurations}. Both the linear coefficients and the orbitals in the configurations are variational parameters, leading to a highly nonlinear problem (with respect to variations of the
orbitals). Each configuration is itself a sum of Slater determinants whose coefficients are fixed such that the configuration belongs to a chosen symmetry subspace ($J=1$ in the case of the example studied in \cite{KimParMarIndDes-98}).

It was noticed in \cite{IndLinDes-05} that the ``error" is slowly disappearing when the number of determinants is increased. This suggests that the phenomenon is purely nonlinear, and that it has nothing to do with the nonrelativistic limit procedure in itself: it is the model obtained in the limit which does not fulfill the usual symmetry properties of nonrelativistic Quantum Chemistry or Physics models. In the limit, one obtains a nonlinear model for which all the configurations have $J=1$ but they do not necessarily have a fixed orbital angular momentum $L^2$ or a fixed total spin $S^2$. In those cases, the calculated ground state  is not an eigenfunction of $L^2$ or $S^2$.

In a linear model, any nondegenerate ground state automatically has the symmetry of the system but in a nonlinear model there could be a \emph{symmetry breaking} phenomenon: although the system has a certain symmetry, the ground state does not necessarily possess this symmetry. One then obtains several minimizers and it is only the set of all ground states which is invariant under the symmetry group.

\medskip

In this paper we study in detail the nonlinear model obtained in the nonrelativistic limit for $Be$-like atoms in the symmetry $J=1$, when only $s$, $p$ and $d$ shells are considered\footnote{That is, we only consider the lowest ``ungerade'' states in the symmetry $J=1$.}, following \cite{KimParMarIndDes-98}. Such an atom has four electrons which can only be in the following subshells: $1s_{1/2}$, $2s_{1/2}$, $2p_{1/2}$, $2p_{3/2}$, $3d_{3/2}$ and $3d_{5/2}$. Each subshell has a certain degeneracy but it is only described by one radial orbital. The distinction between the $2p_{1/2}$ and the $2p_{3/2}$ subshells is an artefact of the nonrelativistic limit. In nonrelativistic theories, the same radial orbital is used for  $2p_{1/2}$ and $2p_{3/2}$. The same holds for $3d_{3/2}$ and $3d_{5/2}$. The fact that the radial orbitals of $2p_{1/2}$ and $2p_{3/2}$ (or $3d_{3/2}$ and $3d_{5/2}$) are allowed to be different is, loosely speaking, similar to enlarging the variational set, which leads to symmetry breaking as will be explained below.

In MCDF theory one considers all the possible configurations of the symmetry $J=1$ which can be constructed upon these states for four electrons, and one optimizes both the radial orbitals of the subshells and the variational coefficients in front of the configurations. Among all these possible wavefunctions, only few of them are eigenfunctions of both $L^2$ and $S^2$. In the present case, there are three possible symmetries: $\sP_1$, $\tP_1$ and $\tD_1$ (corresponding to specific eigenvalues of $L^2$ and $S^2$, see Section \ref{sec:spd} for a precise definition). In Theorem \ref{thm:symmetric_fn} we give necessary and sufficient conditions on the radial orbitals and the coefficients for a wavefunction to be in one of these three symmetries: the radial orbitals of $2p_{1/2}$ and $2p_{3/2}$ (and of $3d_{3/2}$ and $3d_{5/2}$) must be the same and certain relations must hold between the configuration coefficients.

The issue discussed in \cite{KimParMarIndDes-98,FroeseFischer-00,IndLinDes-05} is whether `the' ground state\footnote{We put `the' in quotation marks to emphasize that there is no uniqueness.} of the symmetry $J=1$ is an eigenfunction of $L^2$ and $S^2$. In other words, does it belong to one of the previous symmetries? To address this question, we look at the ground states obtained by \emph{imposing} each of the symmetries $\sP_1$, $\tP_1$ and $\tD_1$, and ask ourselves whether these states can be stationary points and/or local minima of the full model where only $J=1$ is imposed, or not. Our results (Theorems \ref{thm:1P1}, \ref{thm:3D1} and \ref{thm:3P1}) are summarized in the following table:
\begin{center}
\begin{tabular}{|l|l|}
\hline
$\tD_1$ & is a stationary point, but is not a local minimum.\\
\hline
$\sP_1$ & is a stationary point, but is not a local minimum.\\
\hline
$\tP_1$ & has lower energy than $\tD_1$ and $\sP_1$, but is not a stationary point.\\
\hline
\end{tabular}
\end{center}

This shows in particular that enlarging the variational set by only fixing $J=1$ instead of both $L^2$ and $S^2$ leads to symmetry breaking for the wavefunction: the obtained new ground state is never an eigenvector of $L^2$ or $S^2$. It is however essential that all $s$, $p$ and $d$ orbitals are included in the model (similar results are expected when more shells are added, by the same arguments as the ones presented in the paper). As we will see in Theorem \ref{thm:sp}, there is no symmetry breaking for the ground state if only $s$ and $p$ orbitals are considered. 

As noticed first in \cite{KimParMarIndDes-98}, the question whether the constrained ground states of the symmetries $\sP_1$, $\tP_1$ and $\tD_1$ are stationary points of the full model or not, is related to certain properties of the occupation numbers of the orbitals. This is explained in detail in Remarks \ref{rmk:occ1} and \ref{rmk:occ2} below. 

The appearance of symmetry breaking is a well-known phenomenon in atomic multiconfiguration methods; it was encountered and rigorously examined in \cite{CanGalLew-06}. Similar issues occur in molecular calculations with regard to spacial symmetry, see, e.g.,  \cite{DavBor-83,EngLiu-83,McLLenPacEll-85,EisMor-00}. 
In this paper we do not propose any practical solution to this phenomenon. Löwdin who emphasized this issue in a famous discussion \cite{LykPra-63} (after stability results \cite{Slater-30b,CouFis-49,Slater-51,Thouless-60,Adams-62} in Hartree-Fock theory) called it a ``symmetry dilemma''. Our impression is that it is inherent to the way calculations are currently done. For the model studied in this paper, the only reasonable solution is probably to increase the number of determinants in order to decrease the effect of nonlinearities. 

We have tried to make our work accessible to both the Mathematics, Quantum Physics and Chemistry communities. In particular, we will state and prove some well-known results (like for instance a special case of the first Hund's rule \cite{KutMor-96}) for the convenience of the mathematical reader. We hope that our work will help in improving the exchanges between the different communities. On the one hand mathematicians should find the material allowing her/him to understand the models and the typical problems encountered in nonlinear quantum computations. On the other hand we hope physicists and chemists will appreciate our rigorous clarification of the phenomenon they have to deal with.

\section*{Notation}
We recall that the angular momentum reads ${\mathbf L}=x\times (-i\nabla)$, that ${\mathbf S}_k=\sigma_k/2$ where $\sigma_k$ are the well known Pauli matrices, and that ${\mathbf J}={\mathbf L}+{\mathbf S}$. For an $N$-body system, we still denote by ${\mathbf L}=\sum_{k=1}^Nx_k\times (-i\nabla_{x_k})$ the (vector-valued) angular momentum operator of the $N$ particles. A state will be denoted as $\,^{2S+1}\!L_J$ when it is an eigenfunction of $S^2$ with eigenvalue $S(S+1)$, of $L^2$ with eigenvalue $L(L+1)$ (with the identification $P$, $D$, $F$,\,... for $L=1,2,3,\,...$), and of $J^2$ with eigenvalue $J(J+1)$. We will use the notation $\,^{2S+1}\!L$ when it is an eigenfunction of $S^2$ and $L^2$ (with the same eigenvalues as before), but not necessarily an eigenfunction of $J^2$. For more details, we refer to \cite{LandauLifchitz,Thaller-04,Slater-60}.

\section{Model with $sp$ configurations only}
We consider Be-like atoms, i.e. atoms with 4 electrons, in the symmetry $J=1$. We start with a simplified multiconfiguration method employing only $s$ and $p$ shells. This means the $4$-body wavefunction takes the form \cite{KimParMarIndDes-98}
$$\Psi'= a\Phi'(1s_{1/2}^2\, 2s_{1/2}\,2p_{1/2})(R_0,R_1,R_2)+ b\Phi'(1s_{1/2}^2\, 2s_{1/2}\,2p_{3/2})(R_0,R_1,R_3)$$
where $a$ and $b$ are configuration-mixing coefficients and $R_0$, $R_1$, $R_2$ and $R_3$ are the radial functions of, respectively, the shells $1s_{1/2}$, $2s_{1/2}$, $2p_{1/2}$ and $2p_{3/2}$. The two configurations $\Phi'(1s_{1/2}^2\, 2s_{1/2}\,2p_{1/2})$ and $\Phi'(1s_{1/2}^2\, 2s_{1/2}\,2p_{3/2})$ are some linear combination of Slater determinants made upon the corresponding $4$-component shells, with fixed coefficients chosen such that $J=1$. We will not write the detailed form of the configurations here, but later we will give the precise expression of their nonrelativistic limits (see \eqref{relations_s12_p12} and \eqref{relations_s12_p32}). The Hamiltonian to be used is the Dirac $4$-body Coulomb operator which reads
$$H':=\sum_{i=1}^4\left(D^0_i-\frac{Z}{|x_i|}\right)+\sum_{1\leq i<j\leq 4}\frac{1}{|x_i-x_j|},$$
where
$D^0=c\alp\cdot p+mc^2\beta$
is the Dirac operator \cite{Thaller}.

Most atomic MCDF packages  aim at calculating a certain critical point of the energy $(R_0,R_1,R_2,R_3,a,b)\mapsto \pscal{\Psi,H\Psi}$, under the following constraints:
\begin{equation}
\int_0^\ii |R_0(r)|^2\,r^2dr=\int_0^\ii |R_1(r)|^2\,r^2dr=\int_0^\ii |R_2(r)|^2\,r^2dr=\int_0^\ii |R_3(r)|^2\,r^2dr=1,
\label{constraint_sp_1}
\end{equation}
\begin{equation}
\int_0^\ii \overline{R_0(r)}R_1(r)\,r^2dr=0,
\label{constraint_sp_2}
\end{equation}
\begin{equation}
a^2+b^2=1.
\label{constraint_sp_3}
\end{equation}

As the energy is not bounded below due to the negative spectrum of the Dirac operator, it is \emph{a priori} not at all obvious which critical point has to be considered and calculated numerically. Its existence in the infinite dimensional setting is also not clear at all.
However, using the methods of \cite{EstSer-99,EstSer-01}, one can prove that well-chosen critical points of this energy  converge as $c\to\ii$ to critical points of a certain nonrelativistic model which we will now describe in detail.

The variational set of $4$-body wavefunctions obtained in the nonrelativistic limit contains all the functions having the form
\begin{equation}
\Psi= a\Phi(1s_{1/2}^2\, 2s_{1/2}\,2p_{1/2})(R_0,R_1,R_2)+ b\Phi(1s_{1/2}^2\, 2s_{1/2}\,2p_{3/2})(R_0,R_1,R_3)
\label{form_wavefn_sp}
\end{equation}
but this time the two configurations are $2$-component functions, i.e. they only depend on the spin variable. Their relation with the usual nonrelativistic configurations are given by the following formula (see, e.g., \cite{CondonShortley-63} page 294)
\begin{equation}
\Phi(1s_{1/2}^2\, 2s_{1/2}\,2p_{1/2})(R_0,R_1,R_2):=
-\frac{1}{\sqrt{3}}{\,^1\!P_{sp}}(R_0,R_1,R_2)+\frac{\sqrt{2}}{\sqrt{3}}{\,^3\!P_{sp}}(R_0,R_1,R_2), 
\label{relations_s12_p12}
\end{equation}
\begin{equation}
\Phi(1s_{1/2}^2\, 2s_{1/2}\,2p_{3/2})(R_0,R_1,R_3):=\frac{\sqrt{2}}{\sqrt{3}}{\,^1\!P_{sp}}(R_0,R_1,R_3)+\frac{1}{\sqrt{3}}{\,^3\!P_{sp}}(R_0,R_1,R_3). 
\label{relations_s12_p32}
\end{equation}
Here ${\,^1\!P_{sp}}$ and ${\,^3\!P_{sp}}$ are some nonrelativistic configurations chosen such that
$$\left(S^2-k(k+1)\right)\,^{2k+1}\!P_{sp}=0,\qquad \left(L^2-2\right)\,^{2k+1}\!P_{sp}=0\quad\text{ and }\quad \left(J^2-2\right)\,^{2k+1}\!P_{sp}=0$$
for $k=0,1$, and which are made only of $s$ and $p$ orbitals. We should probably rather use the notation $\,^{2k+1}\!P_{1,sp}$ but we refrain to do so, for the sake of simplicity.
The form for the functions ${\,^1\!P_{sp}}$ and ${\,^3\!P_{sp}}$ (with $J_z=1$) is provided in Appendix A for the convenience of the reader. For many of our arguments, we will not need this explicit form.

The associated minimization principle reads
\begin{equation}
E_{sp}(J=1):=\inf_{\substack{R_0,R_1,R_2,R_3,a,b\\ \text{satisfying \eqref{constraint_sp_1}, \eqref{constraint_sp_2} and \eqref{constraint_sp_3}}}}\pscal{\Psi,H\Psi}
\label{def_min_sp}
\end{equation}
where $\Psi$ takes the  form  \eqref{form_wavefn_sp} and $H$ is the nonrelativistic Hamiltonian
$$H:=\sum_{i=1}^4\left(\frac{-\Delta_i}{2}-\frac{Z}{|x_i|}\right)+\sum_{1\leq i<j\leq 4}\frac{1}{|x_i-x_j|}.$$

Note that the Hamiltonian $H$ is real, hence each eigenfunction can be chosen to be real. For this reason, when passing to the nonlinear case, although in principle $H$ should act on complex functions, we will restrict ourselves to real mixing coefficients $\,a,b\,$ and  real-valued radial orbitals $R_k$. The same is done in most Quantum Chemistry or Physics packages. The extension of our arguments to complex functions does not present any difficulty.

The Hamiltonian $H$ commutes with both the total orbital angular momentum $\mathbf{L}$ and the total spin $\mathbf{S}$. For this reason, in a usual  \emph{nonrelativistic} multiconfiguration method, one always restricts the search to ground states of a certain symmetry class. The only configurations made of $s$ and $p$ orbitals satisfying $J=1$ are the ones corresponding to $L=1$ and $S=0$ (singlet) or $S=1$ (triplet), and which appear in \eqref{relations_s12_p12} and \eqref{relations_s12_p32}. Therefore, we will compare \eqref{def_min_sp} with the usual nonrelativistic methods described by the following variational problems:
\begin{equation}
E_{sp}({\,^1\!P_1}):=\inf_{\substack{R_0,R_1,R_2\\ \text{satisfying \eqref{constraint_sp_1} and \eqref{constraint_sp_2}}}}\pscal{{\,^1\!P_{sp}}(R_0,R_1,R_2),H{\,^1\!P_{sp}}(R_0,R_1,R_2)}, 
\label{def_min_sp_1P}
\end{equation}
\begin{equation}
E_{sp}({\,^3\!P_1}):=\inf_{\substack{R_0,R_1,R_2\\ \text{satisfying \eqref{constraint_sp_1} and \eqref{constraint_sp_2}}}}\pscal{{\,^3\!P_{sp}}(R_0,R_1,R_2),H{\,^3\!P_{sp}}(R_0,R_1,R_2)}. 
\label{def_min_sp_3P}
\end{equation}
Note that there is no configuration-mixing coefficient to optimize in the above minimization problems. For this reason, \eqref{def_min_sp_1P} and \eqref{def_min_sp_3P} should indeed be called \emph{Hartree-Fock} methods.

Our wavefunction \eqref{form_wavefn_sp} is always an eigenfunction corresponding to $L=1$ but it is not necessarily an eigenfunction of $S^2$. However, taking $R_2=R_3$ and choosing $a=b/\sqrt{2}$ (resp. $a=-b\sqrt{2}$), we see that our variational set of functions of the form \eqref{form_wavefn_sp} indeed contains all possible functions ${\,^3\!P_{sp}}(R_0,R_1,R_2)$ (resp. ${\,^1\!P_{sp}}(R_0,R_1,R_2)$). Hence we deduce that
\begin{equation}
\boxed{E_{sp}(J=1)\leq \min\big\{E_{sp}({\,^3\!P_1})\, ,\, E_{sp}({\,^1\!P_1})\big\}.}
\label{inegalite_sp}
\end{equation}
The specific case of (the first) Hund's rule proven below in Theorem \ref{thm:Hund} tells us that indeed $E_{sp}({\,^3\!P_1})<E_{sp}({\,^1\!P_1})$, see Corollary \ref{cor:Hund_sp}. In principle, however, the inequality in \eqref{inegalite_sp} could be strict in which case the minimizer would not be an eigenfunction of $S^2$, but instead a linear combination of $S=0$ and $S=1$ states.

We will see that this problem indeed does not occur for $sp$ mixing, as expressed by the
\begin{theorem}[The nonrelativistic limit for $sp$ is correct]\label{thm:sp}
We have
$$E_{sp}(J=1)=E_{sp}({\,^3\!P_1})<E_{sp}({\,^1\!P_1}).$$
Additionaly, any ground state $\Psi$ for $E_{sp}(J=1)$ satisfies $a=\varepsilon b\sqrt{2}$, $R_3=\varepsilon R_2$ for some $\epsilon=\pm1$.
\end{theorem}

The rest of this section is devoted to the (simple) proof of the above theorem. As we will need it in the following, we start by proving a special case of the well-known (first) Hund's rule (for an excellent discussion of Hund's rules, we refer to  \cite{KutMor-96}, where a result similar to the following one is proved).
\begin{theorem}[Hund's rule for singlet/triplet states]\label{thm:Hund}
For $K\geq2$, $M\geq1$, let $\{g_j\}_{j=1}^{2M}\cup\{f_i\}_{i=1}^K$ be an orthonormal system of $L^2(\R^3,\C)$ and\footnote{Throughout the paper, we use the convention $(f_1\wedge\cdots\wedge f_N)(x_1,...,x_N)=(N!)^{-1/2}\det(f_i(x_j))$.}
$$\Psi_1=\!\!\!\sum_{j=1}^Mc_{j}\left(\bigwedge_{i=1}^Kf_i^\uparrow\wedge f_i^\downarrow\right)\wedge g_{2j-1}^\uparrow\wedge g_{2j}^\downarrow,\quad \Psi_2=\!\!\!\sum_{j=1}^Mc_{j}\left(\bigwedge_{i=1}^Kf_i^\uparrow\wedge f_i^\downarrow\right)\wedge g_{2j-1}^\downarrow\wedge g_{2j}^\uparrow$$
where $f^\tau(x,\sigma)=f(x)\delta_{\tau}(\sigma)$ and $(c_j)\in\C^{M}\setminus\{0\}$. Let
$$H:=\sum_{i=1}^{2(K+1)}(h\otimes I_2)_i+\sum_{1\leq i<j\leq 2(K+1)}V(x_i-x_j)$$
be a Hamiltonian where $h$ is a self-adjoint operator on $L^2(\R^3,\C)$ and $V$ is a positive real function. Then, if $\Psi_1$ and $\Psi_2$ belong to the form domain of $H$, one has
\begin{equation}
\pscal{\Psi_1+\Psi_2,H(\Psi_1+\Psi_2)} < \pscal{\Psi_1-\Psi_2,H(\Psi_1-\Psi_2)}.
\end{equation}
\end{theorem}

\begin{remark}
Let us emphasize that we do not impose any spacial symmetry on the functions $f_i$ and $g_j$.
\end{remark}

\begin{proof}
We have
$$\pscal{\Psi_1-\Psi_2,H(\Psi_1-\Psi_2)}-\pscal{\Psi_1+\Psi_2,H(\Psi_1+\Psi_2)}=-4\Re\pscal{\Psi_1,H\Psi_2}=-4\Re\pscal{\Psi_1,\mathbb{V}\Psi_2}$$ 
where $\mathbb{V}$ is the interaction (two-body) potential involving the function $V$. In the last equality we have used that each Slater determinant appearing in $\Psi_1$ always contains two functions orthogonal with all the functions in any of the Slater determinants of $\Psi_2$, which implies that the one-body term vanishes.
Calculating the two-body term one gets
\begin{equation*}
\pscal{\Psi_1,\mathbb{V}\Psi_2}
=-\iint_{\R^3\times\R^3}V(x-y)\left|\sum_{j=1}^Mc_i\,g_{2j-1}(x) g_{2j}(y)\right|^2dx\,dy
\end{equation*}
and the result follows.
\end{proof}
\begin{corollary}[Hund's rule for $sp$ mixing]\label{cor:Hund_sp}
We have $E_{sp}({\,^3\!P_1})<E_{sp}({\,^1\!P_1})$.
\end{corollary}
\begin{proof}
Using the methods of proof of \cite{LieSim-77,Lions-87,Friesecke-03,Lewin-04a}, one can see that there exists $(R_0,R_1,R_2)$ minimizing $E_{sp}({\,^1\!P_1})$. We now choose and fix these functions. Using formula \eqref{3Psp} given in Appendix A and Property \eqref{eq:simplification_2_sp}, we see that
\begin{multline}
\pscal{\tP_{sp}(R_0,R_1,R_2)\,,\,H\big(\tP_{sp}(R_0,R_1,R_2)\big)}\\
=\frac12\pscal{S^-\big(\tP_{2,sp}(R_0,R_1,R_2)\big)\,,\,HS^-\big(\tP_{2,sp}(R_0,R_1,R_2)\big)} 
\label{simplification_S}
\end{multline}
where the state $\tP_{2,sp}$ is defined in \eqref{3P2sp}. Indeed we have precisely
$$S^-\big(\tP_{2,sp}(R_0,R_1,R_2)\big)=s^\uparrow(R_0)\wedge s^\downarrow(R_0)\wedge\bigg( s^\uparrow(R_1)\wedge p_1^\downarrow(R_2)+s^\downarrow(R_1)\wedge p_1^\uparrow(R_2)\bigg)$$
and, see \eqref{1Psp},
$$\sP_{sp}(R_0,R_1,R_2)=\frac{1}{\sqrt{2}}s^\uparrow(R_0)\wedge s^\downarrow(R_0)\wedge\bigg( s^\uparrow(R_1)\wedge p_1^\downarrow(R_2)-s^\downarrow(R_1)\wedge p_1^\uparrow(R_2)\bigg),$$
the notation being that of Appendix A.
The result follows from \eqref{simplification_S} and Theorem \ref{thm:Hund}.
\end{proof}

We now give the proof of Theorem \ref{thm:sp}:

\begin{proof}
Using formulas \eqref{relations_s12_p12}, \eqref{relations_s12_p32}, \eqref{1Psp} and \eqref{3Psp}, we deduce that any trial wavefunction for \eqref{def_min_sp} can be written:
$$\Psi=\frac{1}{\sqrt{3}}{\,^1\!P_{sp}}(R_0,R_1,-aR_2+b\sqrt{2}R_3)+\frac{1}{\sqrt{3}}{\,^3\!P_{sp}}(R_0,R_1,a\sqrt{2}R_2+bR_3).$$ 
Since $H$ commutes with $S$, the scalar product between the above two eigenfunctions of $S^2$ (corresponding to different eigenvalues) vanishes, and we get
\begin{align}
\pscal{\Psi,H\Psi}&=\frac{1}{3}\pscal{{\,^1\!P_{sp}}(R_0,R_1,-aR_2+b\sqrt{2}R_3),H{\,^1\!P_{sp}}(R_0,R_1,-aR_2+b\sqrt{2}R_3)}\nonumber\\ 
&\qquad\qquad+\frac{1}{3}\pscal{{\,^3\!P_{sp}}(R_0,R_1,a\sqrt{2}R_2+bR_3),H{\,^3\!P_{sp}}(R_0,R_1,a\sqrt{2}R_2+bR_3)}\nonumber\\ 
&\geq \frac{\norm{-aR_2+b\sqrt{2}R_3}^2}{3}E_{sp}({\,^1\!P_1})+\frac{\norm{a\sqrt{2}R_2+bR_3}^2}{3}E_{sp}({\,^3\!P_1})\nonumber\\ 
&\geq \frac{\norm{-aR_2+b\sqrt{2}R_3}^2}{3}\big\{E_{sp}({\,^1\!P_1})-E_{sp}({\,^3\!P_1})\big\}+E_{sp}({\,^3\!P_1})\geq E_{sp}({\,^3\!P_1})\label{eq:estim_sp_OK}
\end{align}
where $||\cdot ||$ denotes the $L^2(r^2\,dr)$ norm and we have used that
$$\norm{-aR_2+b\sqrt{2}R_3}^2+\norm{a\sqrt{2}R_2+bR_3}^2=3(a^2+b^2)=3.$$
By Hund's rule (Corollary \ref{cor:Hund_sp}), we have $E_{sp}({\,^1\!P_1})-E_{sp}({\,^3\!P_1})>0$, hence we get that $E_{sp}(J=1)\geq E_{sp}({\,^3\!P_1})$. Therefore there must be equality in \eqref{eq:estim_sp_OK} and it holds $aR_2=b\sqrt{2}R_3$.
Taking the square of the previous relation and using that $\int (R_2)^2=\int (R_3)^2$, we prove  the result.
\end{proof}

\section{Model with $sp$ and $pd$ configurations}\label{sec:spd}
In the previous section we have seen that the nonrelativistic limit of our model with only $sp$ configurations was ``correct". We now study in detail the model which was considered in \cite{KimParMarIndDes-98}. The idea is to add configurations by considering $d$ shells. The nonrelativistic wavefunction now takes the form:
\begin{multline}
\Psi= a\Phi(1s_{1/2}^2\, 2s_{1/2}\,2p_{1/2})(R_0,R_1,R_2)+ b\Phi(1s_{1/2}^2\, 2s_{1/2}\,2p_{3/2})(R_0,R_1,R_3)\\
+
c \Phi(1s_{1/2}^2\, 2p_{1/2}\,3d_{3/2})(R_0,R_2,R_4)+d\Phi(1s_{1/2}^2\, 2p_{3/2}\,3d_{3/2})(R_0,R_3,R_4)\\
+e\Phi(1s_{1/2}^2\, 2p_{3/2}\,3d_{5/2})(R_0,R_3,R_5).
\label{form_wavefn_sp_pd}
\end{multline}
Here functions $R_i$ are  accounting for the radial part of each shell orbital, which are normalized like in \eqref{constraint_sp_2}. Only $R_0$ and $R_1$ have to be orthogonal. As before one can first consider the relativistic model with $4$-component wavefunctions and pass to the nonrelativistic limit $c\to\ii$. One obtains the above form \eqref{form_wavefn_sp_pd} of the wavefunction. The first two functions of \eqref{form_wavefn_sp_pd} have already been defined in \eqref{relations_s12_p12} and \eqref{relations_s12_p32}. The other three functions are given by (see again \cite{CondonShortley-63} page 294)
\begin{multline}
\Phi(1s_{1/2}^2\, 2p_{1/2}\,3d_{3/2})(R_0,R_2,R_4)
:=\frac{1}{\sqrt{3}}{\,^1\!P_{pd}}(R_0,R_2,R_4)-\frac{1}{\sqrt{6}}{\,^3\!P_{pd}}(R_0,R_2,R_4)\\ 
+\frac{1}{\sqrt{2}}{\,^3D_{pd}}(R_0,R_2,R_4),
\label{relations_p12_d32}
\end{multline}
\begin{multline}
\Phi(1s_{1/2}^2\, 2p_{3/2}\,3d_{3/2})(R_0,R_3,R_4)
:=-\frac{1}{\sqrt{15}}{\,^1\!P_{pd}}(R_0,R_3,R_4)+\frac{2\sqrt{2}}{\sqrt{15}}{\,^3\!P_{pd}}(R_0,R_3,R_4)\\ 
+\frac{\sqrt{2}}{\sqrt{5}}{\,^3D_{pd}}(R_0,R_3,R_4),
\label{relations_p32_d32}
\end{multline}
\begin{multline}
\Phi(1s_{1/2}^2\, 2p_{3/2}\,3d_{5/2})(R_0,R_3,R_5)
:=\frac{\sqrt{3}}{\sqrt{5}}{\,^1\!P_{pd}}(R_0,R_3,R_5)+\frac{\sqrt{3}}{\sqrt{10}}{\,^3\!P_{pd}}(R_0,R_3,R_5)\\ 
-\frac{1}{\sqrt{10}}{\,^3D_{pd}}(R_0,R_3,R_5).
\label{relations_p32_d52}
\end{multline}
The above nonrelativistic configurations satisfy, for $k=0,1$
$$\left(S^2-k(k+1)\right)\,^{2k+1}\!P_{pd}=\left(L^2-2\right)\,^{2k+1}\!P_{pd}=\left(J^2-2\right)\,^{2k+1}\!P_{pd}=0,$$ 
$$\left(S^2-2\right)\,^{3}\!D_{pd}=\left(L^2-6\right)\,^{3}\!D_{pd}=\left(J^2-2\right)\,^{3}\!D_{pd}=0.$$ 
Formulas for these nonrelativistic functions of the $pd$ shells with $J_z=1$ are given in Appendix A.

\subsection{Eigenfunctions of $L^2$ and $S^2$}
Among functions of the form \eqref{form_wavefn_sp_pd}, we will be interested in the ones which are eigenfunctions of $S^2$ and $L^2$, i.e. the ones which have the symmetry which is imposed in nonrelativistic calculations. Note that, contrarily to the $sp$ mixing studied in the previous section, our wavefunction is \emph{a priori} not even an eigenfunction of $L^2$.

We will write $\Psi\in \,^1\!P_1$ when $\Psi$ is a linear combination of configurations $\,^1\!P_{sp}$ and $\,^1\!P_{pd}$. We use similar notations for $\Psi\in\,^3\!P_1$ and $\Psi\in\,^3D_1$. The following result will be crucial in our analysis:
\begin{theorem}[Eigenvectors of $L^2$ and $S^2$ of the form \eqref{form_wavefn_sp_pd}]\label{thm:symmetric_fn} Let $\Psi$ a normalized wavefunction of the form \eqref{form_wavefn_sp_pd}.
\smallskip
\begin{enumerate}
\item We have $\Psi\in \,^3\!P_1$ if and only if there exists $\varepsilon,\varepsilon' =\pm1$ such that $R_2=\varepsilon \,R_3$, $R_4=\varepsilon'\,R_5$, $a=\varepsilon \sqrt{2}\,b$, $3d=4\varepsilon'e$ and $4c=-\varepsilon \sqrt{5}\,d$.
In this case
\begin{equation}
\Psi = \frac{a}{\sqrt{3}}\left(\sqrt{2} + \frac{1}{\sqrt{2}} \right)\,{}^3\!P_{sp}\,(R_0,R_1, R_2) - c\sqrt{6}\,\,{}^3\!P_{pd}\,(R_0,R_2, R_4),
\label{IND-triplet}
\end{equation}
\item We have $\Psi\in \,^1\!P_1$ if and only if there exists $\varepsilon,\varepsilon' =\pm1$ such that $R_2=\varepsilon \,R_3$, $R_4=\varepsilon'\,R_5$ $a\sqrt{2}=-\varepsilon \,b$, $3d=-\varepsilon'e$ and $c=-\varepsilon \sqrt{5}\,d$.
In this case
\begin{equation}
\Psi = -a\sqrt{3}\,{}^1\!P_{sp}(R_0,R_1, R_2) + c\,\sqrt{3}\,\,{}^1\!P_{pd}(R_0,R_2, R_4).
\label{IND-singlet}
\end{equation}
\item We have $\Psi\in \,^3D_1$ if and only if there exist $\varepsilon,\varepsilon',\varepsilon'' =\pm1$ such that $a=b=0$, $R_2=\varepsilon R_3$, $R_4=\varepsilon'\,R_5$, $c=\varepsilon''/\sqrt{2}$, $d=\sqrt{2/5}\,\varepsilon\varepsilon''$, $e=-\sqrt{1/10}\,\varepsilon\varepsilon'\varepsilon''$.
In this case
\begin{equation}
\Psi = \varepsilon''\;{}^3\!D_{pd}(R_0,R_2, R_4).
\label{IND-D}
\end{equation}
\end{enumerate}
\end{theorem}

\begin{proof}
Using formulas \eqref{form_wavefn_sp_pd}--\eqref{relations_p32_d52}, we get (for the sake of clarity, we omit to mention $R_0$ which appears in all configurations)
\begin{multline}
\Psi= \frac{1}{\sqrt{3}}\;{}^3\!P_{sp}(R_1, a\sqrt{2} R_2+b R_3)
+\frac1{\sqrt{30}}\left(
\,{}^3\!P_{pd}(-c\sqrt{5}R_2+4d R_3, R_4) +
\,{}^3\!P_{pd}(3e R_3, R_5)\right)\\
+\frac{1}{\sqrt{3}}\;{}^1\!P_{sp}(R_1,-a R_2+b\sqrt{2}R_3)+
\frac1{\sqrt{30}}\left(
\,{}^1\!P_{pd}(c\sqrt{10}R_2-d\sqrt{2}R_3, R_4) +
\,{}^1\!P_{pd}(3e\sqrt{2}R_3, R_5)\right)\\
+ \frac1{\sqrt{30}}\left(
\,{}^3D_{pd}(c\sqrt{15}R_2+2\sqrt{3}d R_3 , R_4)+\;{}^3D_{pd}(-e\sqrt{3}R_3,  R_5)
\right).
\label{formula_Psi}
\end{multline}
Hence, using the orthogonality properties of the different configurations, we see that $\Psi\in\,^3\!P_1$ if and only if
$$\left\{\begin{array}{rl}
-a \,R_2+b\sqrt{2}\,\,R_3&=0\,,\\
c\sqrt{5}R_2+(3e -d)\,R_3&= 3e(1-(R_4,R_5))\,R_3\,,\\
c\sqrt{5}\,R_2+(2d -e)\,R_3 &= -e (1-(R_4,R_5))\,R_3\,,\\
e (R_5-(R_4,R_5)R_4) &=0\,.
\end{array}\right.$$
The last equation tells us that either $e=0$ or $R_4=\epsilon'R_5$ with $\epsilon'=\pm1$.
If $e=0$ then the second and third equations imply that $c=d=0$ and only the first equation remains.
If $R_4=R_5$, then the system reduces to
$$\left\{\begin{array}{rl}
-a \,R_2+b\sqrt{2}\,\,R_3&=0\,,\\
c\sqrt{5}R_2+(3e -d)\,R_3&=0\,,\\
c\sqrt{5}\,R_2+(2d -e) \,R_3 &= 0\,,\\
R_4&=R_5\,.
\end{array}\right.$$
The second and third equations then imply that $3d=4e$. The rest follows from the normalization of $R_2$ and $R_3$.
The proof is similar for $\epsilon'=-1$, and for $\,^1\!P_1$ and $\,^3D_1$ states.
\end{proof}

In view of the above result, we now introduce the nonrelativistic ground state energies with $sp+pd$ mixing
\begin{equation}
E_{sp+pd}({\,^3\!P_1}):=\inf_{\Psi \text{ of the form \eqref{IND-triplet}}}\pscal{\Psi,H\Psi},
\label{def_min_sppd_3P}
\end{equation}
\begin{equation}
E_{sp+pd}({\,^1\!P_1}):=\inf_{\Psi \text{ of the form \eqref{IND-singlet}}}\pscal{\Psi,H\Psi},
\label{def_min_sppd_1P}
\end{equation}
\begin{equation}
E_{pd}({\,^3\!D_1}):=\inf_{\Psi \text{ of the form \eqref{IND-D}}}\pscal{\Psi,H\Psi}.
\label{def_min_sppd_3D}
\end{equation}
Our goal is to compare these nonrelativistic energies with the one obtained in the nonrelativistic limit:
\begin{equation}
E_{sp+pd}(J=1):=\inf_{\Psi \text{ of the form \eqref{form_wavefn_sp_pd}}}\pscal{\Psi,H\Psi}.
\label{def_min_sp_pd}
\end{equation}
By definition, we of course have
\begin{equation}
\boxed{E_{sp+pd}(J=1)\leq \min\bigg\{E_{sp+pd}({\,^3\!P_1})\,,\,E_{sp+pd}({\,^1\!P_1})\,,\,E_{sp+pd}({\,^3\!D_1})\bigg\}.} 
\label{ineg_var_pb}
\end{equation}
The phenomenon which was observed by Kim \emph{et al.} in \cite{KimParMarIndDes-98} was precisely that `the' ground state for $E_{sp+pd}(J=1)$ (i.e. the nonrelativistic limit of `the' MCDF ground state) was \emph{not} an eigenfunction of $L$ and $S$, hence it was not a solution of any of the problems $E_{sp+pd}({\,^3\!P_1})$, $E_{sp+pd}({\,^1\!P_1})$ or $E_{sp+pd}({\,^3\!D_1})$. This means that there must be a strict inequality $<$ in \eqref{ineg_var_pb}. This relaxation phenomenon is itself the reason for the deficiency of the nonrelativistic limit of MCDF theory. It is a typical nonlinear phenomenon.

We will now study with more details the three minimizers for \eqref{def_min_sppd_1P}, \eqref{def_min_sppd_3D} and \eqref{def_min_sppd_3P}. We will in particular give some conditions under which \eqref{ineg_var_pb} is a strict inequality, and we will check these conditions numerically in Section \ref{sec:num}.

\subsection{Study of the $\sP_1$ state}
In this section, we prove that `the' $\sP_1$ state minimizing $E_{sp+pd}({\,^1\!P_1})$ is \emph{never} a ground state for $J=1$, although it is always a stationary point of the associated energy functional.
\begin{theorem}[The $\sP_1$ state]\label{thm:1P1} Let $\Psi\in\sP_1$ be a function of the form \eqref{IND-singlet} minimizing $E_{sp+pd}({\,^1\!P_1})$ defined in \eqref{def_min_sppd_1P}.
\begin{enumerate}
\item The associated mixing coefficients and orbitals $(a,...,e,R_0,...,R_5)$ satisfying the relations of Theorem \ref{thm:symmetric_fn} (Assertion $2$) provide a \emph{stationnary point} of the total energy functional
$$(a,...,e,R_0,...,R_5)\mapsto \pscal{\Psi,H\Psi}$$
where $\Psi$ takes the form \eqref{form_wavefn_sp_pd}.
\item However, $\Psi$ is \emph{never} a local minimum of this functional.
\end{enumerate}
\end{theorem}
\begin{proof}
We first show that $\Psi$ is a critical point of the total energy functional. For simplicity we assume that $c\neq0$, the proof being the same otherwise. We also assume for simplicity that $\epsilon=\epsilon'=1$, hence $R_2=R_3$ and $R_4=R_5$. We have to consider both variations with respect to mixing coefficients, and to orbitals. The vanishing of the variation with respect to the mixing coefficients $(a,...,e)$ is a simple consequence of the fact that there is no overlap between states belonging to different symmetry spaces,
\begin{equation*}
\pscal{H\Psi,\tP_{sp}(...)}=\pscal{H\Psi,\tP_{pd}(...)}=\pscal{H\Psi,\tD_{pd}(...)}=0, 
\end{equation*}
for any radial functions in the corresponding states.

We now turn to the variation with respect to orbitals. By extremality of the singlet function $\Psi$ among $\sP_1$ states and due to the constraints on $R_0,...,R_5$, we have\footnote{The derivatives appearing below are the coordinates of the gradient of the energy with respect to the scalar product of $L^2([0,\ii),r^2\,dr)$.}
\begin{equation}\label{deriv01-simplet}
\frac{\partial{\mathcal E}(\Psi)}{\partial R_i}\in\text{span}(R_0,R_1) \;\mbox{ for }\; i=0,1\,,
\end{equation}
\begin{equation}\label{deriv23-simplet}
\frac{\partial{\mathcal E}(\Psi)}{\partial R_2}+\frac{\partial{\mathcal E}(\Psi)}{\partial R_3}{{|_{R_3=R_2}}}\in \text{span}(R_2)
\end{equation}
and
\begin{equation}\label{deriv45-simplet}
\frac{\partial{\mathcal E}(\Psi)}{\partial R_4}+\frac{\partial{\mathcal E}(\Psi)}{\partial R_5}{{|_{R_5=R_4}}}\in\text{span}(R_4).
\end{equation}

Using  $\epsilon=1,\, a\,\sqrt{2}=- b$, $3d=-e$ and $c=-\sqrt{5}\,d$ as given by Theorem \ref{thm:symmetric_fn}, we see that at $\Psi$ we have
\begin{equation}\label{happy1}\frac{\partial{\mathcal E}(\Psi)}{\partial R_2}=\frac12\,\frac{\partial{\mathcal E}(\Psi)}{\partial R_3}{{|_{R_3=R_2}}}\,,\end{equation}
and
\begin{equation}\label{happy2}\frac{\partial{\mathcal E}(\Psi)}{\partial R_4}=\frac23\,\frac{\partial{\mathcal E}(\Psi)}{\partial R_5}{{|_{R_5=R_4}}}\,,
\end{equation}
From (\ref{deriv23-simplet})-(\ref{happy1}) and (\ref{deriv45-simplet})-(\ref{happy2}), we find
\begin{equation}
\label{deriv05-simplet}\frac{\partial{\mathcal E}(\Psi)}{\partial R_i}\in\text{span}(R_i) \;\mbox{ for }\; i=2,...,5
\end{equation}
which ends the proof of the criticality of $\Psi$.

\begin{remark}\label{rmk:occ1}
The exceptional relation \eqref{happy1} holds true because, as was noticed first in \cite{KimParMarIndDes-98}, for the singlet state the ratio between the occupation numbers of the $p_{1/2}$ function $R_2$ and the occupation number of the $p_{3/2}$ function $R_3$ in the $\sP_{sp}$ state, is the same as the corresponding ratio for the $\sP_{pd}$ state. Let us explain this with more details. Considering a variation $\delta R\in(R_2=R_3)^\perp$, we find that the variations of the total energy functional are, using \eqref{formula_Psi},
$$\frac{\partial{\mathcal E}(\Psi)}{\partial R_2}(\delta R)=2\pscal{H\Psi\, ,\, \left(\frac{-a}{\sqrt{3}}\sP_{sp}(R_1,\delta R)+\frac{c\sqrt{10}}{\sqrt{30}}\sP_{pd}(\delta R,R_4)\right)},$$
$$\frac{\partial{\mathcal E}(\Psi)}{\partial R_2}(\delta R)=2\pscal{H\Psi\, ,\, \left(\frac{b\sqrt{2}}{\sqrt{3}}\sP_{sp}(R_1,\delta R)+\frac{-d\sqrt{2}+3e\sqrt{2}}{\sqrt{30}}\sP_{pd}(\delta R,R_4)\right)}.$$
When the matrix
$$\left(\begin{matrix}
\frac{-a}{\sqrt{3}} & \frac{c\sqrt{10}}{\sqrt{30}}\\
\frac{b\sqrt{2}}{\sqrt{3}} & \frac{-d\sqrt{2}+3e\sqrt{2}}{\sqrt{30}}
\end{matrix}\right)$$
is \emph{not} invertible, its columns are colinear and we have $\frac{\partial{\mathcal E}(\Psi)}{\partial R_2}(\delta R)=
k\,\frac{\partial{\mathcal E}(\Psi)}{\partial R_3}(\delta R)$ for some $k$ ($k=1/2$ in our case). Hence
$$\frac{\partial{\mathcal E}(\Psi)}{\partial R_2}(\delta R)+\frac{\partial{\mathcal E}(\Psi)}{\partial R_3}(\delta R)=0\Longleftrightarrow \left\{\begin{array}{rl}
\displaystyle\frac{\partial{\mathcal E}(\Psi)}{\partial R_2}(\delta R)&=0,\\
\displaystyle\frac{\partial{\mathcal E}(\Psi)}{\partial R_3}(\delta R)&=0,\\
\end{array}\right.$$
as we want. This holds true when
$$\frac{\frac{-a}{\sqrt{3}}}{\frac{b\sqrt{2}}{\sqrt{3}}}=\frac{\frac{c\sqrt{10}}{\sqrt{30}}}{\frac{-d\sqrt{2}+3e\sqrt{2}}{\sqrt{30}}}.$$ 
Using the relations between $a,...,e$ provided by Theorem \ref{thm:symmetric_fn}, we see that the above equality reduces to (when $a,c\neq0$)
$$\frac{\frac13}{\frac23}=\frac{\frac{10}{30}}{\frac{20}{30}}$$
which is precisely the ratio between the occupation numbers as mentioned before.
The argument is the same for $R_4=R_5$.
\end{remark}

We now turn to the proof that $\Psi$ is never a local minimum, which will simply follow from Hund's rule. Note that we can write the relations between $a,b,...,e$ as
$$\left(\begin{matrix} a\\ b\end{matrix}\right)=-a\sqrt{3}\,U_{sp}^*\left(\begin{matrix} 0\\ 1\end{matrix}\right),\qquad\left(\begin{matrix} c\\ d\\ e\end{matrix}\right)=c\sqrt{3}\,U_{pd}^*\left(\begin{matrix} 0\\ 0\\1\end{matrix}\right)$$
where $U_{sp}$ and $U_{pd}$ are the Condon-Shortley unitary matrices \cite{CondonShortley-63}
$$U_{sp}=\frac{1}{\sqrt{3}}\left(\begin{matrix} \sqrt{2} & 1\\ -1&\sqrt{2}\end{matrix}\right),\qquad U_{pd}=\frac{1}{\sqrt{30}}\left(\begin{matrix} -\sqrt{3} & 3 & 3\sqrt{2}\\ 2\sqrt{3} & 4 & -\sqrt{2}\\ \sqrt{15} & -\sqrt{5} & \sqrt{10}\end{matrix}\right).$$
Now we define the following new mixing coefficients
$$\left(\begin{matrix} a'\\ b'\end{matrix}\right)=-a\sqrt{3}\,U_{sp}^*\left(\begin{matrix} 1\\ 0\end{matrix}\right),\qquad\left(\begin{matrix} c'\\ d'\\ e'\end{matrix}\right)=c\sqrt{3}\,U_{pd}^*\left(\begin{matrix} 0\\ 1\\0\end{matrix}\right)$$
and note that by construction $(a',...,e')$ is orthogonal to $(a,...,e)$ for the scalar product of $\R^5$. Also we have
\begin{align*}
\Psi'&:=\Psi(a',b',c',d',e',R_0,R_1,R_2=R_3,R_4=R_5)\\
&=-a\sqrt{3}\,{}^3\!P_{sp}(R_0,R_1, R_2) + c\,\sqrt{3}\,\,{}^3\!P_{pd}(R_0,R_2, R_4),
\end{align*}
i.e. it takes exactly the same form as $\Psi$ but with triplet states instead of singlet states. Now we vary the mixing coefficients as follows $\sqrt{1-t^2}(a,...,e)+t(a',...,e')$, which results into a variation for the wavefunction of the form $\sqrt{1-t^2}\Psi+t\Psi'$. Calculating the energy of this new wavefunction we find
$$\cE(\sqrt{1-t^2}\Psi+t\Psi')=\cE(\Psi)+t^2\left(\cE(\Psi')-\cE(\Psi)\right).$$ 
Note that there is no first order term since $\Psi$ is a stationary state as shown before (or simply because $\Psi$ and $\Psi'$ belong to different symmetry spaces, hence $\pscal{H\Psi,\Psi'}=0$).

Now we claim that $\cE(\Psi')<\cE(\Psi)$, which will clearly imply that $\Psi$ cannot be a local minimum. By \eqref{eq:simplification_2_sp} and \eqref{eq:simplification_2_pd} in Appendix A, we know that $\cE(\Psi')=\cE( \Psi'')$ where
$$\Psi''=-\frac{a\sqrt{3}}{\sqrt{2}}S^-\tP_{2,sp}(R_0,R_1, R_2) + \frac{c\,\sqrt{3}}{\sqrt{2}}S^-\tP_{2,pd}(R_0,R_2, R_4)$$
which is the simple triplet state taking the same form as $\Psi$ but with the adequate signs reversed. The last step is to apply Theorem \ref{thm:Hund}, with the following functions: $f_1=s(R_0)$, $g_1=p_1(R_2)$,  $g_2=(c^2/20+a^2)^{-1/2}\big(c/(2\sqrt{5})d_0(R_4)+as(R_1)\big)$, $g_3=p_{-1}(R_2)$, $g_4=d_2(R_4)$, $g_5=p_0(R_2)$ and $g_6=d_1(R_4)$.
This ends the proof of Theorem \ref{thm:1P1}.
\end{proof}

\begin{remark}\label{rmk:eigenvalues}
When $R_2=\pm R_3$ and $R_4=R_5$, the problem consisting of varying only the mixing coefficient essentially reduces to that of finding the eigenvalues of the nonrelativistic Hamiltonian matrix, i.e. the matrix of $H$ in the space spanned by the 5 configurations built upon the orbitals:
$$\sP_{sp}(R_0,R_1,R_2),\ \sP_{pd}(R_0,R_2,R_4),\ \tP_{sp}(R_0,R_1,R_2),\ \tP_{pd}(R_0,R_2,R_4),\ \tD_{pd}(R_0,R_2,R_4).$$
This $5\times5$ matrix is block diagonal and its eigenvalues are:
$$\lambda_1(\,^k\!P_1)=\inf_{\alpha^2+\beta^2=1}\cE\left(\alpha\,^k\!P_{sp}(R_0,R_1,R_2)+\beta \,^k\!P_{pd}(R_0,R_2,R_4)\right)$$
$$\lambda_2(\,^k\!P_1)=\sup_{\alpha^2+\beta^2=1}\cE\left(\alpha\,^k\!P_{sp}(R_0,R_1,R_2)+\beta \,^k\!P_{pd}(R_0,R_2,R_4)\right)$$
$$\lambda(\tD_1)=\pscal{H\tD_{pd}(R_0,R_2,R_4),\tD_{pd}(R_0,R_2,R_4)}.$$
What we have used in this second part is that $\lambda_1(\,^3\!P_1)<\lambda_1(\,^1\!P_1)$, by Hund's rule (indeed, we even have that the $2\times2$ Hamiltonian matrix of $\tP_1$ states is smaller than the one of $\sP_1$ states, in the sense of quadratic forms). However although it is expected that in many cases $\lambda_1(\,^1\!P_1)<\lambda(\,^3\!D_1)$, there is no general rule: this may depend on the orbitals $R_0,...,R_5$.
\end{remark}

\subsection{Study of the $\tD_1$ state}
In this section, we prove that `the' $\tD_1$ state minimizing $E_{pd}({\,^3\!D_1})$ is also always a stationary point of the associated energy functional and we give a condition implying that it is not a local minimum.
\begin{theorem}[The $\tD_1$ state]\label{thm:3D1} Let $\Psi\in\tD_1$ be a function of the form \eqref{IND-D} minimizing $E_{pd}({\,^3\!D_1})$ defined in \eqref{def_min_sppd_3D}.
\begin{enumerate}
\item The associated mixing coefficients and orbitals $(a,...,e,R_0,...,R_5)$ given by Theorem \ref{thm:symmetric_fn} provide a \emph{stationary point} of the total energy functional
\begin{equation}
(a,...,e,R_0,...,R_5)\mapsto \pscal{\Psi,H\Psi}
\label{functional_3D}
\end{equation}
where $\Psi$ takes the form \eqref{form_wavefn_sp_pd}.
\item If moreover
\begin{equation}
\inf_{\substack{\alpha^2+\beta^2=1\\ R_1\in\{R_0\}^\perp,\ \norm{R_1 }_{L^2(r^2dr)}=1}}\cE\left(\alpha\tP_{sp}(R_0,R_1,R_2)+\beta\tP_{pd}(R_0,R_2,R_4)\right)<E_{pd}({\,^3\!D_1}), 
\label{condition_3D1}
\end{equation}
where $R_0$, $R_2$, $R_4$ are the radial functions of $\Psi$, then the stationary state $\Psi$ is not a local minimum of the functional \eqref{functional_3D}.
\end{enumerate}
\end{theorem}

\begin{remark}
Our assumption \eqref{condition_3D1} exactly means that $\lambda(\tD_1)=E_{pd}({\,^3\!D_1})$ is not the lowest eigenvalue of the Hamiltonian matrix as explained above in Remark \ref{rmk:eigenvalues}. 
It is not clear how to prove  \eqref{condition_3D1} rigorously. In Section \ref{sec:num} we will verify it  numerically on the approximated solutions provided by the ATSP HF and MCHF program by Froese-Fischer \cite{ATSP}.
\end{remark}

\begin{proof}
We do not give all the details of the proof which is very similar to that of the $\sP_1$ state. The fact that $\Psi$ is a stationary point is seen exactly as in Theorem \ref{thm:1P1}.

The proof that $\Psi$ is not a local minimum is also very similar to that of the $\sP_1$ state, with the difference that we do not have a general Hund's rule for $\tD_1$ states, hence we need to require condition \eqref{condition_3D1}, which  is expected to be true in many cases. This time we have (assuming again $\epsilon=\epsilon'=\varepsilon''=1$ for simplicity)
$$\left(\begin{matrix} a\\ b\end{matrix}\right)=\left(\begin{matrix} 0\\ 0\end{matrix}\right),\qquad\left(\begin{matrix} c\\ d\\ e\end{matrix}\right)=\,U_{pd}^*\left(\begin{matrix} 1\\ 0\\0\end{matrix}\right).$$
The result is then obtained by arguing as before with, this time,
$$\left(\begin{matrix} a'\\ b'\end{matrix}\right)=\alpha\,U_{sp}^*\left(\begin{matrix} 0\\ 1\end{matrix}\right),\qquad\left(\begin{matrix} c'\\ d'\\ e'\end{matrix}\right)=\beta\,U_{pd}^*\left(\begin{matrix} 0\\ 1\\0\end{matrix}\right),$$
where $\alpha$ and $\beta$ are chosen to minimize the left side of \eqref{condition_3D1}.
\end{proof}

\subsection{Study of the $\tP_1$ state}
In this section, we give a simple condition implying that `the' $\tP_1$ state minimizing $E_{sp+pd}({\,^3\!P_1})$ is \emph{not} a stationary point of the total energy functional.
\begin{theorem}[The $\tP_1$ state]\label{thm:3P1} Let $\Psi\in\tP_1$ be a function of the form \eqref{IND-triplet} minimizing $E_{sp+pd}({\,^3\!P_1})$ defined in \eqref{def_min_sppd_3P}, and denote by
$(a,...,e,R_0,...,R_5)$ the associated mixing coefficients and orbitals satisfying the relations of Theorem \ref{thm:symmetric_fn}.

If there exists $\delta R\in (R_3)^\perp$ such that
\begin{equation}
c\pscal{H\Psi,\tP_{pd}(R_0,\delta R,R_4)}\neq0\quad \text{or}\quad a\pscal{H\Psi,\tP_{sp}(R_0,R_1,\delta R)}\neq0,
\label{condition_3P1}
\end{equation}
then $\Psi$ does \emph{not} provide a {stationary point} of the total energy functional
$$(a,...,e,R_0,...,R_5)\mapsto \pscal{\Psi,H\Psi}.$$
\end{theorem}

\begin{remark}
Condition \eqref{condition_3P1} is very intuitive. It indeed \emph{implies} that
\begin{equation}
E_{sp+p'd}(\tP_1)<E_{sp+pd}(\tP_1)\,,
\label{ineqq}
\end{equation}
where
$$E_{sp+p'd}(\tP_1)=\inf_{\substack{\alpha^2+\beta^2=1,\\(R_0,...,R_4)\\\text{ satisfying \eqref{constraint_sp_1} and \eqref{constraint_sp_2}} }}\cE\left(\alpha\tP_{sp}(R_0,R_1,R_2)+\beta\tP_{pd}(R_0,R_3,R_4)\right).$$ 
The relation between \eqref{condition_3P1} and \eqref{ineqq}
 was already noticed in \cite{FroeseFischer-00}.

Note that when $a,c\neq0$, one can prove that $ \pscal{H\Psi,\tP_{pd}(R_0,\delta R,R_4)}\neq0$ or that $\pscal{H\Psi,\tP_{sp}(R_0,R_1,\delta R)}\neq0$ for some $\delta R$ and these two conditions are indeed equivalent.
In Section \ref{sec:num}, the condition $\pscal{H\Psi,\tP_{sp}(R_0,R_1,\delta R)}\neq0$ is verified numerically.
\end{remark}

We now give the
\begin{proof}
As before, it can easily be seen that our state $\Psi$ is indeed a stationary state with respect to variations of the mixing coefficients only. The non-criticality will come from the variations of the orbitals, as suggested by \eqref{condition_3P1}. As before we assume for simplicity that $\epsilon=\epsilon'=1$, hence $R_2=R_3$ and $R_4=R_5$ for the ground state $\Psi$.

By extremality of the triplet function $\Psi$ among $\tP_1$ states, we have similarly as before
\begin{equation}\label{deriv01-triplet}
\frac{\partial{\mathcal E}(\Psi)}{\partial R_i}\in\text{span}(R_0,R_1) \;\mbox{ for }\; i=0,1\,,
\end{equation}
\begin{equation}\label{deriv23-triplet}
\frac{\partial{\mathcal E}(\Psi)}{\partial R_2}+\frac{\partial{\mathcal E}(\Psi)}{\partial R_3}{{|_{R_3=R_2}}}\in\text{span}(R_2)\
\end{equation}
and
\begin{equation}\label{deriv45-triplet}
\frac{\partial{\mathcal E}(\Psi)}{\partial R_4}+\frac{\partial{\mathcal E}(\Psi)}{\partial R_5}{{|_{R_5=R_4}}}\in\text{span}(R_4).
\end{equation}
This time we find using the relations of Theorem \ref{thm:symmetric_fn} that
$$\frac{\partial{\mathcal E}(\Psi)}{\partial R_4}=\frac{\partial{\mathcal E}(\Psi)}{\partial R_5}{{|_{R_5=R_4}}},$$
hence only variations with respect to $R_2$ and $R_3$ remain to be considered.
The main point is that there is \emph{a priori} no relation between $\frac{\partial{\mathcal E}(\Psi)}{\partial R_2}$ and $\frac{\partial{\mathcal E}(\Psi)}{\partial R_3}$.
More precisely, let us consider a variation $\delta R\in (R_2=R_3)^\perp$. We have
$$\frac{\partial{\mathcal E}(\Psi)}{\partial R_2}(\delta R)=2\pscal{H\Psi\, ,\, \left(\frac{a\sqrt{2}}{\sqrt{3}}\tP_{sp}(R_1,\delta R)+\frac{-c\sqrt{5}}{\sqrt{30}}\tP_{pd}(\delta R,R_4)\right)}$$
and
\begin{align*}
\frac{\partial{\mathcal E}(\Psi)}{\partial R_3}(\delta R)&=2\pscal{H\Psi\, ,\, \left(\frac{b}{\sqrt{3}}\tP_{sp}(R_1,\delta R)+\frac{4d+3e}{\sqrt{30}}\tP_{pd}(\delta R,R_4)\right)}\\
&=2\pscal{H\Psi\, ,\, \left(\frac{a}{\sqrt{6}}\tP_{sp}(R_1,\delta R)+\frac{-5c\sqrt{5}}{\sqrt{30}}\tP_{pd}(\delta R,R_4)\right)},
\end{align*}
where we have used the relations of Theorem \ref{thm:symmetric_fn} in the last line. The main difference with the $\sP_1$ state is now that the matrix
\begin{equation}
\left(\begin{matrix}\frac{\sqrt{2}}{\sqrt{3}}&-\frac{\sqrt{5}}{\sqrt{30}}\\   \frac{1}{\sqrt{6}} & -\frac{5\sqrt{5}}{\sqrt{30}}     \end{matrix}\right)
\label{matrix_3P1}
\end{equation}
is invertible. Hence we get that $\Psi$ is stationary with respect to variations of $R_2$ and $R_3$ independently if and only if
$$\forall \delta R\in (R_2=R_3)^\perp,\qquad a\pscal{H\Psi\, ,\, \tP_{sp}(R_1,\delta R)}=c\pscal{H\Psi\, ,\, \tP_{pd}(\delta R,R_4)}=0.$$
This clearly leads to a contradiction when \eqref{condition_3P1} holds true.
\end{proof}

\begin{remark}\label{rmk:occ2}
The fact that the matrix \eqref{matrix_3P1} is invertible can be interpreted in saying that the ratio between the occupation number of the $p_{1/2}$ function $R_2$ and the one of the $p_{3/2}$ function $R_3$ in the $sp$ configuration is not the same as the one of the $pd$ configuration \cite{KimParMarIndDes-98}.
\end{remark}

\subsection{Conclusion}
In the previous sections we have studied the states of the nonrelativistic symmetries $\tP_1$, $\sP_1$ and $\tD_1$. As a consequence of our results we obtain the
\begin{corollary}[Occurrence of symmetry breaking for $sp+pd$]
Assume that \eqref{condition_3D1} holds for a $\tD_1$ ground state and that condition \eqref{condition_3P1} holds for a $\tP_1$ ground state. Then we have
\begin{equation}
E_{sp+pd}(J=1)< \min\bigg\{E_{sp+pd}({\,^3\!P_1})\,,\,E_{sp+pd}({\,^1\!P_1})\,,\,E_{sp+pd}({\,^3\!D_1})\bigg\}. 
\label{ineg_strict}
\end{equation}
Additionally a ground state $\Psi$ for $E_{sp+pd}(J=1)$ is never an eigenfunction of $L^2$ and neither of $S^2$.
\end{corollary}

\begin{proof}
The strict inequality \eqref{ineg_strict} is an obvious consequence of the previous results. It implies that a ground state for the minimization problem $E_{sp+pd}(J=1)$ cannot be a common eigenfunction of $L^2$ and $S^2$.
What remains to be proven is that it cannot even be an eigenfunction of $L^2$ or of  $S^2$ separately. This means that any ground state must have a nonvanishing projection in each of the symmetries $\sP_1$, $\tP_1$ and $\tD_1$.

Assume for instance that $\Psi$ is an eigenfunction of $L^2$. It cannot be a $\tD_1$ state by \eqref{ineg_strict}, hence one must have $\Psi\in\sP_1+\tP_1$. Using \eqref{formula_Psi}, one obtains the conditions
$$\left\{\begin{array}{rl}
c\sqrt{15}R_2+\left(2d\sqrt{3}-e\sqrt{3}(R_4,R_5)\right)R_3&=0,\\
e\left(R_5-(R_4,R_5)R_4\right)&=0.
\end{array}\right.$$
If $e=0$ we get $R_2=\epsilon R_3$ and $c\sqrt{5}=-2\epsilon d$ for some $\epsilon\in\{\pm1\}$. Hence
\begin{multline*}
\Psi=\frac{a\sqrt{2}+\epsilon b}{\sqrt{3}}\;{}^3\!P_{sp}(R_1, R_2)
+\frac{-c\sqrt{2}}{\sqrt{3}}
\,{}^3\!P_{pd}(R_2, R_4) \\
+\frac{-a+b\epsilon\sqrt{2}}{\sqrt{3}}\;{}^1\!P_{sp}(R_1,R_2)+
\frac{c\sqrt{3}}{2}\,{}^1\!P_{pd}(R_2, R_4).
\end{multline*}
By Hund's rule we get
\begin{multline*}
\cE(\Psi)\geq \norm{\frac{a\sqrt{2}+\epsilon b}{\sqrt{3}}\;{}^3\!P_{sp}(R_1, R_2)
+\frac{-c\sqrt{2}}{\sqrt{3}}
\,{}^3\!P_{pd}(R_2, R_4)}^2E_{sp+pd}(\tP_1)\\
+\norm{\frac{-a+b\epsilon\sqrt{2}}{\sqrt{3}}\;{}^1\!P_{sp}(R_1,R_2)+
\frac{c\sqrt{3}}{2}\,{}^1\!P_{pd}(R_2, R_4)}^2E_{sp+pd}(\sP_1)\geq E_{sp+pd}(\tP_1)
\end{multline*}
which contradicts \eqref{ineg_strict}. The argument is the same if $R_4=\pm R_5$. Hence we have shown that $\Psi$ cannot be an eigenfunction of $L^2$.
The proof that $\Psi$ cannot be an eigenfunction of $S^2$ is very similar.
\end{proof}

\section{Numerical verification of \eqref{condition_3D1} and \eqref{condition_3P1}}\label{sec:num}

In this section we verify numerically the two assumptions made in the previous section. We use the package ATSP of Froese-Fischer \cite{ATSP} to get approximations of the nonrelativistic states in the different symmetry spaces. The program uses 220 discretization points on a logarithmic grid and  a finite difference method. We have then verified conditions \eqref{condition_3D1} and \eqref{condition_3P1}  with the help of {\sl Mathematica}.

\subsection{Verification of \eqref{condition_3P1}}

First we start with the MCHF calculation for the configuration ${\,^3\!P_{sp+pd}}\,$ for the Beryllium atom. This yields constants $a, c$ and functions $R_0, R_1, R_2, R_4$, which are numerical approximations of the real ones. We can also run the program for the simpler Hartree-Fock cases ${\,^3\!P_{sp}}$, ${\,^3\!P_{pd}}$ and we get the following total energies (in Hartree).
\begin{align*}E_{sp}\big({}^3\!P_1\big)&\simeq -14.5115\,,\\
E_{sp+pd}\big({}^3\!P_1\big)&\simeq -14.5166\,.
\end{align*}
Also we obtain for the ${}^3\!P_{sp+pd}$ configuration the following numerical values for $a$ and $c$:
$$a\simeq 0.9951963\,,\qquad c \simeq -0.0978997\,.$$
The plots of the radial parts of the orbitals $R_0$, $R_1$, $R_2=R_3$ and $R_4=R_5$ (corresponding respectively to the shells $1s$, $2s$, $2p$ and $3d$) are displayed in Figure \ref{fig:1s_2s_2p_3d} below.

\begin{figure}[h]
\begin{center}
\begin{tabular}{cc}
\includegraphics[width=6.5cm]{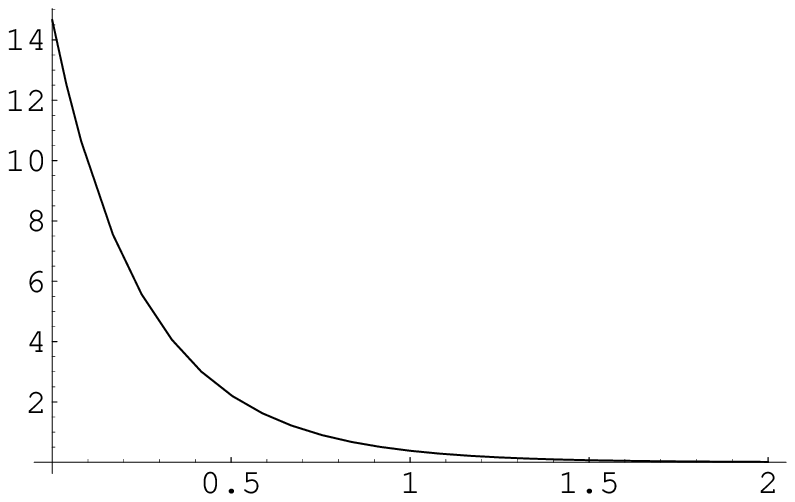} & \includegraphics[width=6.5cm]{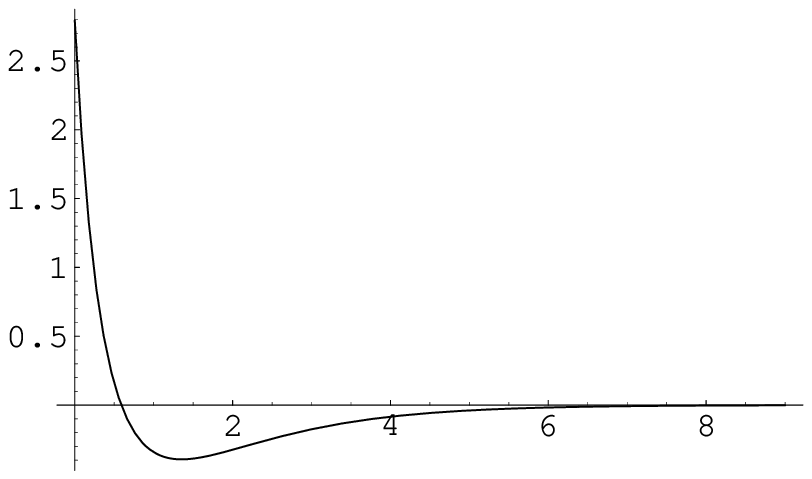}\\
\includegraphics[width=6.5cm]{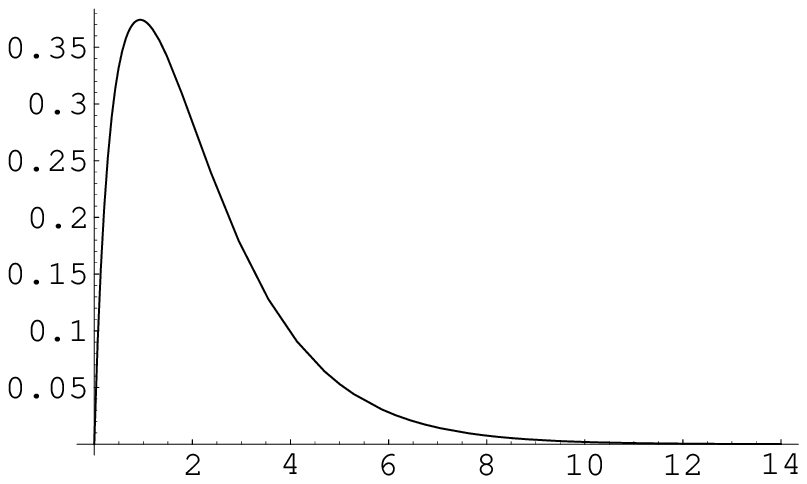} & \includegraphics[width=6.5cm]{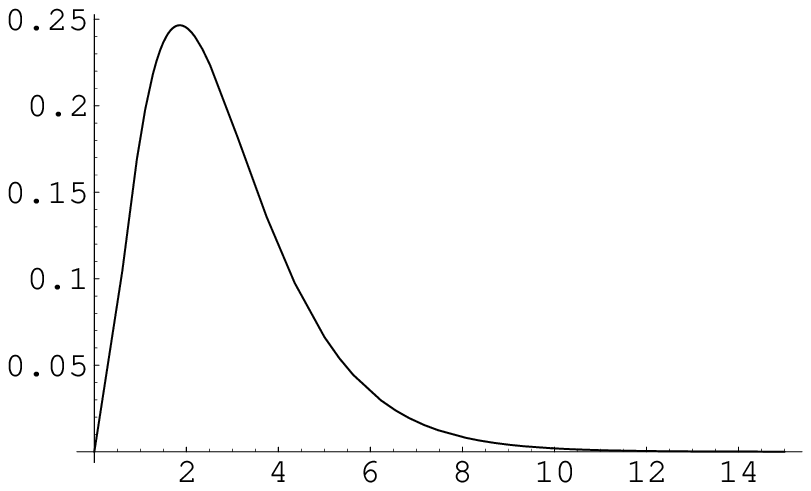}
\end{tabular}
\caption{\label{fig:1s_2s_2p_3d}\it Plots of $R_0$ (top left), $R_1$ (top right), $R_2$ (bottom left) and $R_4$ (bottom right) for the optimal $\,{}^3\!P_1(sp+pd)\,$ configuration of the Beryllium atom, obtained with the package ATSP.}
\end{center}
\end{figure}

Now, we are going to verify \eqref{condition_3P1} by showing the existence of a function $\delta R$, orthogonal to $R_2$, such that 
\begin{equation}
\pscal{H\Psi,\tP_{sp}(R_0,R_1,\delta R)}\neq0\,.
\label{cond_to _verify} 
\end{equation}
As explained in detail in Appendix B, we can calculate the exact expression of \eqref{cond_to _verify} in terms of the radial functions only. We obtain
\begin{align*}
 \pscal{\,H\Psi,\,\tP_{sp}(R_0,R_1,\delta R)\,} &={a}\int_0^\infty \left(\frac{s^4}{2}\left(\frac{R_2(s)}{s}\right)'\left(\frac{\delta R(s)}{s}\right)'-4\,s\,R_2(s)\,\delta R(s)\right)\,ds\\
&\quad+a \int_0^\infty\int_0^\infty\left(2|R_0(s)|^2+|R_1(s)|^2 \right)\,  R_2(t)\, \delta R(t)\,\frac{s^2\,t^2}{\max{\{s,t\}}}\,ds\,dt\\
&\quad-a\,\int_0^\infty\int_0^\infty  R_0(s)\, \delta R(s)\,   R_2(t)\, R_0(t)\;\frac{s^2\,t^2\,\min{\{s,t\}}}{3\,\max^2{\{s,t\}}}\,ds\,dt \\
&\quad-a\int_0^\infty\int_0^\infty  R_1(s)\, \delta R(s)\,   R_2(t)\, R_1(t)\;\frac{s^2\,t^2\,\min{\{s,t\}}}{3\,\max^2{\{s,t\}}}\,ds\,dt \\
& \quad -c\,\sqrt{2}\, \int_0^\infty\int_0^\infty R_1(s)\,   R_2(s)\, \delta R(t)\,   R_4(t)\;\frac{s^2\,t^2\,\min{\{s,t\}}}{3\,\max^2{\{s,t\}}}\,ds\,dt \,.
\end{align*}

Finding a function  $\delta R$, orthogonal to $R_2$, for which the above formula is away from $0$ is less convincing than arguing as follows. The above formula is linear with respect to $\delta R$ hence it can be written
$$\pscal{\,H\Psi,\,\tP_{sp}(R_0,R_1,\delta R)\,}=\int_0^\ii F(r) \delta R(r)\,dr$$
with
\begin{align*}
&F(r):= a\left(\!-\frac{r^2}2\,R_2''(r)-r\,R_2'(r) + R_2(r) - 4r\,R_2(r)\right)
+a \int_0^\infty\!\!\left(2|R_0(s)|^2+|R_1(s)|^2 \right)  \frac{R_2(r)\,s^2\,r^2}{\max{\{s,r\}}}ds\\
& -a\int_0^\infty  R_0(r)\,    R_2(s)\, R_0(s)\;\frac{r^2\,s^2\,\min{\{r,s\}}}{3\,\max^2{\{r,s\}}}\,ds 
-a\int_0^\infty  R_1(r)\,    R_2(s)\, R_1(s)\;\frac{r^2\,s^2\,\min{\{r,s\}}}{3\,\max^2{\{r,s\}}}\,ds \\
&\qquad\qquad\quad  -{c\,\sqrt{2}}\, \int_0^\infty R_1(s)\,   R_2(s)\,    R_4(r)\;\frac{s^2\,r^2\,\min{\{s,r\}}}{3\,\max^2{\{r,t\}}}\,ds \,.
\end{align*}
Then we note that the vanishing of $\int_0^\ii F(r) \delta R(r)\,dr$ for all $\delta R$ orthonormal to $R_2$, is indeed equivalent to the existence of a Lagrange multiplier $\lambda$ such that $F=\lambda R_2$. Hence it suffices to show  that $F(r)/R_2(r)$ is not constant in order to verify \eqref{condition_3P1}.

We compute numerically $F/R_2$ and get the following graph: \\
\begin{figure}[h]\begin{center}\includegraphics[width=6cm]{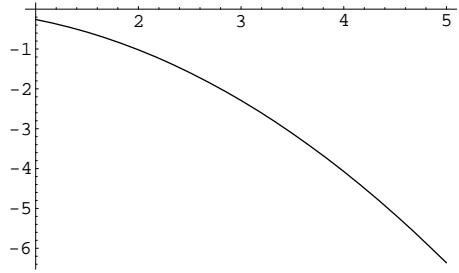}\hspace*{1.5cm}\caption{\it Plot of $F(r)/R_2(r)$ for $r$ between $1$ and $5$.\,}\end{center}\end{figure}\\
This clearly shows that the function $F/R_2$ is not constant. Hence \eqref{condition_3P1} holds true.

\subsection{Verification of \eqref{condition_3D1}}

Let us now pass to the numerical verification of \eqref{condition_3D1}. In order to do so, we again run the Froese-Fischer ATSP program \cite{ATSP} for the HF calculation of the ${\,^3\!D_{ pd}}$ configuration for the Beryllium atom. This provides us with radial functions $R_0$,   $R_2$ and $R_4$ for the orbitals $1s, 2p, 3d$. For the $2s$ orbital, we construct a function $R_1$ by taking the one obtained in the calculation for the ${\,^3\!P_{ sp+pd}}$ configuration and projecting it on the orthogonal to the space generated by $\,  R_0\,$. We get the following numbers for the total energies (in Hartree):
\begin{align*}
E_{pd}({}^3\!D_{1})&\simeq  -14.1889\,,\\
\cE\left({}^3\!P_{sp}(  R_0,   R_1,   R_2)\right)&\simeq -14.4998\,,
\end{align*}
thus yielding a numerical verification of condition \eqref{condition_3D1}.

\appendix
\section*{Appendix A. Nonrelativistic configurations in the sector $J=1$ and $J_z=1$}
\setcounter{section}{1}
\addcontentsline{toc}{section}{Appendix A. Nonrelativistic configurations in the sector $J=1$ and $J_z=1$}
\label{app:nonrelativistic_cfg}

For the convenience of the reader, we quickly explain how to construct the nonrelativistic configurations used in the text.
In the whole appendix we use the notation
$$s^{\tau}(R)(x,\sigma)=R(|x|)\delta_{\tau}(\sigma),\qquad p^{\tau}_m(R)(x,\sigma)=R(|x|)Y_1^m\left(\frac{x}{|x|}\right)\delta_{\tau}(\sigma),$$ 
$$d^{\tau}_m(R)(x,\sigma)=R(|x|)Y_2^m\left(\frac{x}{|x|}\right)\delta_{\tau}(\sigma)$$ 
for $\tau\in\{\uparrow,\downarrow\}$ and $m=-\ell,...,\ell$. Here $Y^m_\ell$ are the usual eigenfunctions of $L^2$ and $L_z$ (orbital angular momentum), normalized in the Hilbert space $L^2(S^2)$, such that
$$L^2Y^m_\ell=\ell(\ell+1)Y^m_\ell,\qquad L_zY^m_\ell=mY^m_\ell,$$
and which can be expressed in terms of Legendre polynomials \cite{CondonShortley-63}.

\subsection*{$sp$ configurations}
\begin{table}
\centering
\includegraphics[width=13cm]{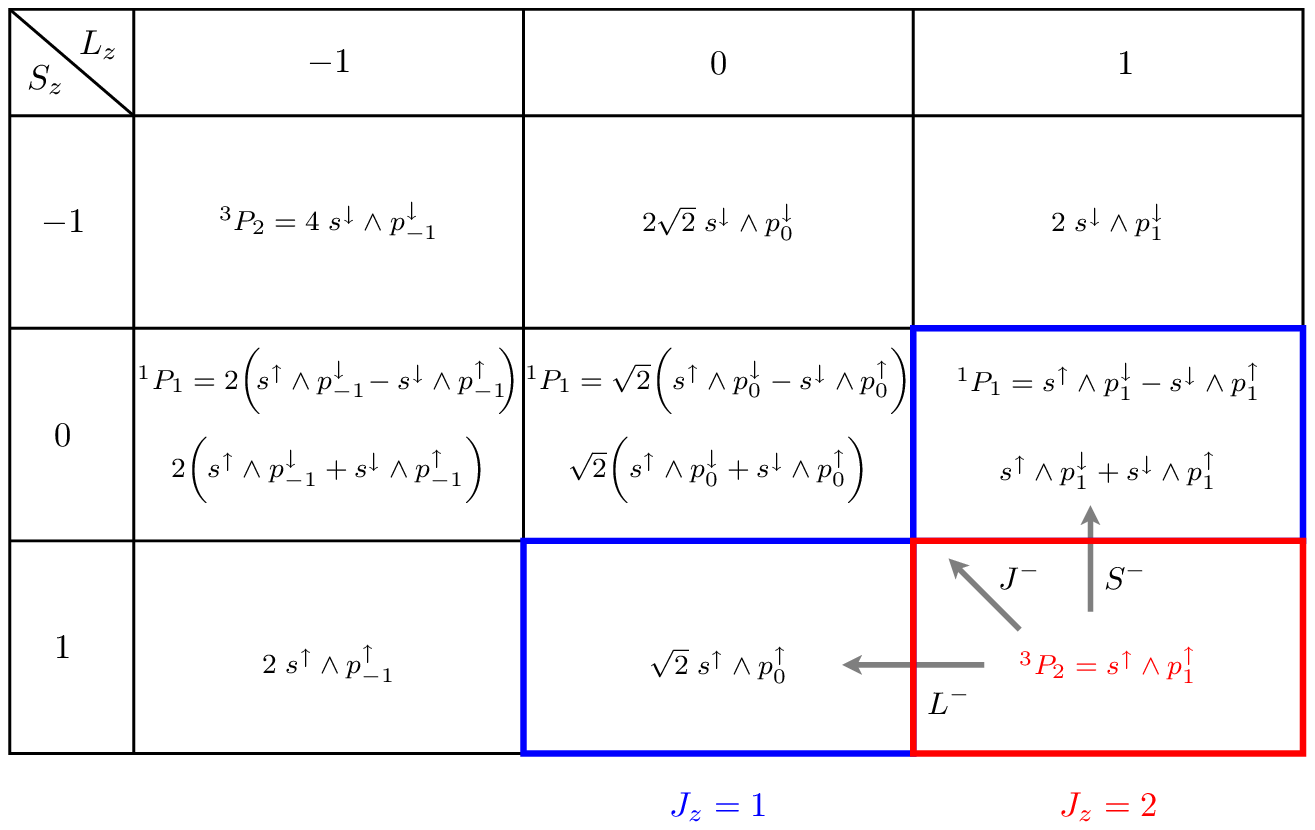}
\caption{Derivation of all the common eigenvectors of $L^2$, $S^2$, $L_z$ and $S_z$, starting from the $\tP_{2,sp}$ state having  all the highest possible quantum numbers, at the bottom right of the table. For the sake of clarity, we have not indicated the closed shell $s_0^\uparrow(R_0)\wedge s_0^\downarrow(R_0)\wedge\cdots$ which appears in front of all the configurations. We have also refrained from mentioning the radial functions $R_1$ and $R_2$.}
\label{tab:sp}
\end{table}

We start with $sp$ configurations. Table \ref{tab:sp} summarizes the different possible states which are eigenfunctions of both $L^2$ and $S^2$, classified with respect to their $L_z$ and $S_z$. The usual method is to start from the state having the highest possible quantum numbers ($S=1$, $S_z=1$, $L=1$, $L_z=1$, $J=2$ and $J_z=2$) which is located at the lower right corner:
\begin{equation}\label{3P2sp}
\tP_{2,sp}(R_0,R_1,R_2):=s^\uparrow(R_0)\wedge s^\downarrow(R_0)\wedge\bigg( s^\uparrow(R_1)\wedge p_1^\uparrow(R_2)\bigg).
\end{equation}

Next we apply the lowering operators $L^-=L_x-iL_y$ and $S^-=S_x-iS_y$ to get two states having the same $S=1$ and $L=1$ but without a precise $J$:
\begin{equation}
S^-\big(\tP_{2,sp}(R_0,R_1,R_2)\big)=s^\uparrow(R_0)\wedge s^\downarrow(R_0)\wedge\bigg( s^\uparrow(R_1)\wedge p_1^\downarrow(R_2)+s^\downarrow(R_1)\wedge p_1^\uparrow(R_2)\bigg),
\label{S_ref_sp}
\end{equation}
\begin{equation}
L^-\big(\tP_{2,sp}(R_0,R_1,R_2)\big):=\sqrt{2}\;s^\uparrow(R_0)\wedge s^\downarrow(R_0)\wedge\bigg( s^\uparrow(R_1)\wedge p_0^\uparrow(R_2)\bigg).
\label{L_ref_sp}
\end{equation}
A usual antisymmetrization of \eqref{S_ref_sp} gives our (normalized) singlet state having $S=0$, hence $J=L=1$:
\begin{equation}
{\,^1\!P_{sp}}(R_0,R_1,R_2)=\frac{1}{\sqrt{2}}s^\uparrow(R_0)\wedge s^\downarrow(R_0)\wedge\bigg( s^\uparrow(R_1)\wedge p_1^\downarrow(R_2)-s^\downarrow(R_1)\wedge p_1^\uparrow(R_2)\bigg).
\label{1Psp}
\end{equation}
In the subspace $J_z=1$, we can construct a $\tP_{2}$ state by applying the lowering operator $J^-=L^-+S^-$:
$$J^-\big(\tP_{2,sp}(R_0,R_1,R_2)\big)=L^-\big(\tP_{2,sp}(R_0,R_1,R_2)\big)+S^-\big(\tP_{2,sp}(R_0,R_1,R_2)\big).$$ 
We know that the subspace $J_z=1$ is of dimension 3. Hence we deduce by orthogonality that our (normalized) $\tP_1$ state with $J=1$ and $J_z=1$ is
\begin{equation}
{\,^3\!P_{sp}}(R_0,R_1,R_2)=\frac{(L^--S^-)}2\tP_{2,sp}(R_0,R_1,R_2).
\label{3Psp}
\end{equation}
Note that
\begin{equation}
L^+L^-\big(\tP_{2,sp}(R_0,R_1,R_2)\big)=2\big(\tP_{2,sp}(R_0,R_1,R_2)\big), 
\label{LL_sp}
\end{equation}
\begin{equation}
S^+S^-\big(\tP_{2,sp}(R_0,R_1,R_2)\big)=2\big(\tP_{2,sp}(R_0,R_1,R_2)\big), 
\label{SS_sp}
\end{equation}
which explains why $1/2$ is the right normalization in \eqref{3Psp}. From \eqref{LL_sp} and \eqref{SS_sp} we see that for any observable $A$ commuting with $L$ and $S$ (for instance $A=H$, our nonrelativistic Hamiltonian), then we have for all $a,b$ and all $R_0,R_1,R_2$ and $R_0',R_1',R_2'$
\begin{align}
&\bigg\langle (aL^-+bS^-)\big(\tP_{2,sp}(R_0,R_1,R_2)\big)\,,\,
A\;(aL^-+bS^-)\big(\tP_{2,sp}(R'_0,R'_1,R'_2)\big)\bigg\rangle_{L^2(\R^3)^4}\nonumber\\ 
& \qquad\qquad =(a^2+b^2)\bigg\langle L^-\big(\tP_{2,sp}(R_0,R_1,R_2)\big)\,,\,
A\;L^-(\tP_{2,sp}(R'_0,R'_1,R'_2)\big)\bigg\rangle_{L^2(\R^3)^4}\nonumber\\
& \qquad\qquad =(a^2+b^2)\bigg\langle S^-\big(\tP_{2,sp}(R_0,R_1,R_2)\big)\,,\,
A\;S^-\big(\tP_{2,sp}(R'_0,R'_1,R'_2)\big)\bigg\rangle_{L^2(\R^3)^4}\nonumber\\ 
& \qquad\qquad =2(a^2+b^2)\bigg\langle \tP_{2,sp}(R_0,R_1,R_2)\,,\,
A\;\tP_{2,sp}(R'_0,R'_1,R'_2)\bigg\rangle_{L^2(\R^3)^4}.\label{eq:simplification_2_sp} 
\end{align}

\subsection*{$pd$ configurations}
We now switch to the calculation of the $pd$ configurations.

\begin{table}
\centering
\includegraphics[width=14cm]{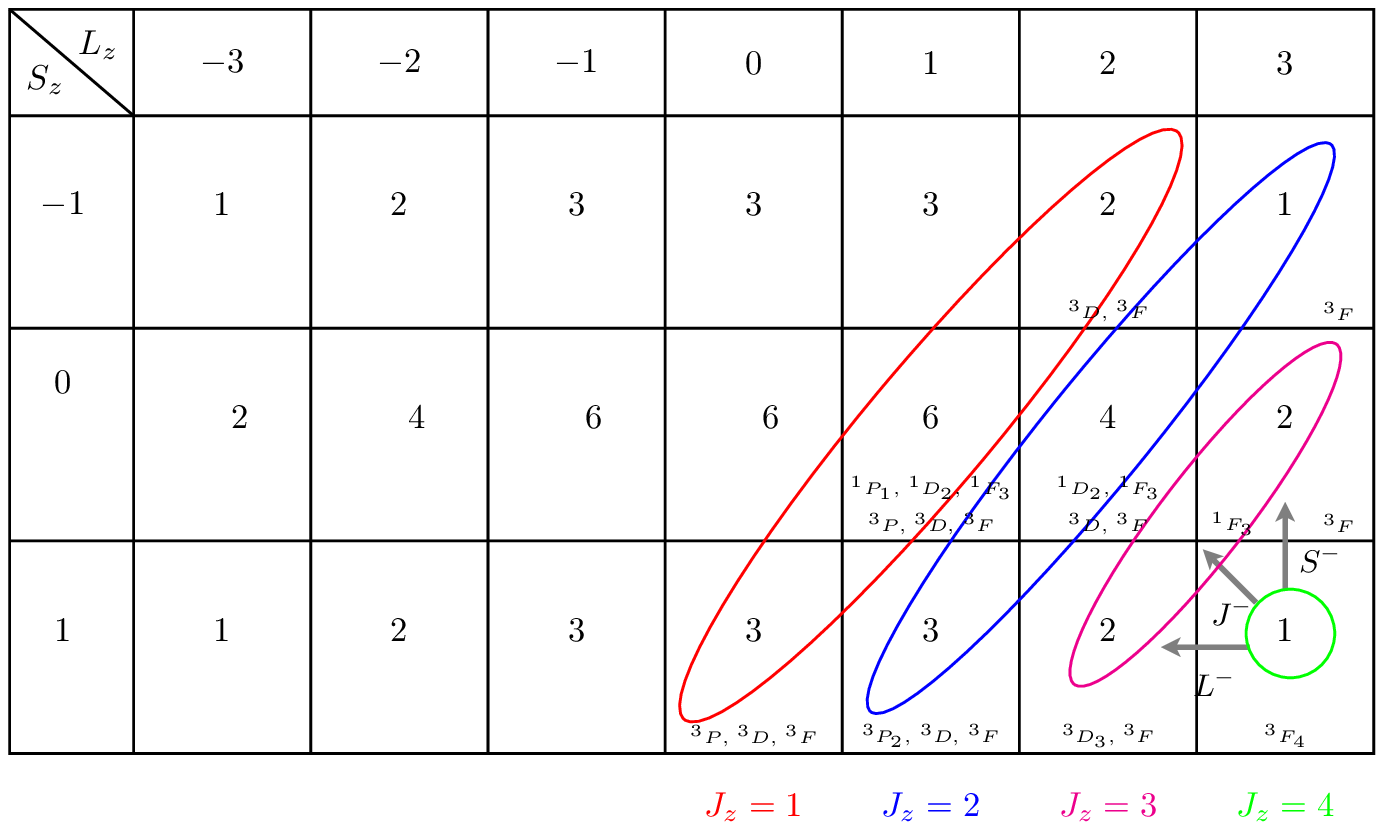}
\caption{Dimensions of the different common eigenspaces of $L_z$ and $S_z$ for $pd$ configurations. In each cell, we have indicated the type of state spanning the corresponding space.}
\label{tab:pd}
\end{table}

Table \ref{tab:pd} shows the dimension of the different common eigenspaces of $L_z$ and $S_z$ for $pd$ configurations. As before we may start from the $^3F_4$ state having the highest quantum numbers and derive our states in the $J=1$, $J_z=1$ sector by applying successively the lowering operators $L^-$, $S^-$ and $J^-$.

The final result is as follows. The $\sP_1$ state in the $J_z=1$ sector reads
\begin{multline}
{\,^1\!P_{pd}}(R_0,R_1,R_2)=\frac{1}{2\sqrt{5}}s^\uparrow(R_0)\wedge s^\downarrow(R_0)\wedge\bigg( \sqrt{6}p_{-1}^\uparrow(R_1)\wedge d_2^\downarrow(R_2)-\sqrt{3}p_{0}^\uparrow(R_1)\wedge d_1^\downarrow(R_2)\\
+p_{1}^\uparrow(R_1)\wedge d_0^\downarrow(R_2)-\sqrt{6}p_{-1}^\downarrow(R_1)\wedge d_2^\uparrow(R_2)+\sqrt{3}p_{0}^\downarrow(R_1)\wedge d_1^\uparrow(R_2)
-p_{1}^\downarrow(R_1)\wedge d_0^\uparrow(R_2)\bigg).
\label{1Ppd}
\end{multline}
The $\tP_1$ state reads 
\begin{equation}
{\,^3\!P_{pd}}(R_0,R_1,R_2)=\frac{L^--S^-}{2} \tP_{2,pd}(R_0,R_1,R_2)
\label{3Ppd}
\end{equation}
where
\begin{multline}
\tP_{2,pd}(R_0,R_1,R_2)= s^\uparrow(R_0)\wedge s^\downarrow(R_0)\wedge\bigg(\frac{\sqrt{6}}{\sqrt{10}}p_{-1}^\uparrow(R_1)\wedge d_{2}^\uparrow(R_2)-\frac{\sqrt{3}}{\sqrt{10}}p_{0}^\uparrow(R_1)\wedge d_{1}^\uparrow(R_2)\\
+\frac{1}{\sqrt{10}}p_{1}^\uparrow(R_1)\wedge d_{0}^\uparrow(R_2)\bigg)\label{def:3Ppd_2}
\end{multline}
is the unique $\tP_2$ configuration having $J_z=2$. Note that, like in the $sp$ case, one has
\begin{equation}
L^+L^-\big(\tP_{2,pd}(R_0,R_1,R_2)\big)=2\big(\tP_{2,pd}(R_0,R_1,R_2)\big), 
\label{LL_pd}
\end{equation}
\begin{equation}
S^+S^-\big(\tP_{2,pd}(R_0,R_1,R_2)\big)=2\big(\tP_{2,pd}(R_0,R_1,R_2)\big). 
\label{SS_pd}
\end{equation}
Therefore we deduce that \eqref{eq:simplification_2_sp} still holds for linear combinations of the triplets $sp$ and $pd$ and that for any observable $A$ commuting with $L$ and $S$,
\begin{align}
&\bigg\langle a\tP_{sp}(R_0,R_1,R_2)+b\tP_{pd}(R_0,R_3,R_4)\,,\,
A\;\big(a'\tP_{sp}(R'_0,R'_1,R'_2)+b'\tP_{pd}(R'_0,R'_3,R'_4)\big)\bigg\rangle_{L^2(\R^3)^4}\nonumber\\ 
&\qquad\qquad=\frac{1}{2}\bigg\langle aS^-\big(\tP_{2,sp}(R_0,R_1,R_2)\big)+bS^-\big(\tP_{2,pd}(R_0,R_3,R_4)\big)\,,\nonumber\\
&\qquad\qquad\qquad\qquad\,
A\;\big(a'S^-\big(\tP_{2,sp}(R_0',R_1',R'_2)\big)+b'S^-\big(\tP_{2,pd}(R'_0,R'_3,R'_4)\big)\big)\bigg\rangle_{L^2(\R^3)^4}\\ 
&\qquad\qquad=\bigg\langle a\tP_{2,sp}(R_0,R_1,R_2)+b\tP_{2,pd}(R_0,R_3,R_4)\,,\nonumber\\
&\qquad\qquad\qquad\qquad\,
A\;\big(a'\tP_{2,sp}(R_0',R_1',R'_2)+b'\tP_{2,pd}(R'_0,R'_3,R'_4)\big)\bigg\rangle_{L^2(\R^3)^4}.\label{eq:simplification_2_pd} 
\end{align}
Lastly the $\tD_1$ state is obtained through the more complicated formula
\begin{equation}
\tD_{pd}(R_0,R_1,R_2)=\frac{1}{6\sqrt{2}}\big(L^-L^--{3}L^-S^-+{6}S^-S^-\big)\tD_{3,pd}(R_0,R_1,R_2) 
\label{3Dpd}
\end{equation}
where $\tD_{3,pd}(R_0,R_1,R_2)$ is the unique $\tD_3$ state having $J_z=3$, given by
\begin{equation}
\tD_{3,pd}(R_0,R_1,R_2)=s^\uparrow(R_0)\wedge s^\downarrow(R_0)\wedge\bigg(\frac{\sqrt{2}}{\sqrt{3}}p_{0}^\uparrow(R_1)\wedge d_{2}^\uparrow(R_2)-\frac{1}{\sqrt{3}}p_{1}^\uparrow(R_1)\wedge d_{1}^\uparrow(R_2)\bigg).
\end{equation}

\section*{Appendix B. Energy expressions as functions of the radial components}
\setcounter{section}{2}
\addcontentsline{toc}{section}{Appendix B. Energy expressions as functions of the radial components}
\label{app:integrals}
In this appendix, we provide the formulas of the energy written in terms of the radial components of the orbitals. We will first need the
\begin{lemma}\label{4to2}
Assume that we have 6  mutually orthogonal functions, $f_1,\dots, f_6$. Then,
\begin{equation}
\pscal{f_1\wedge f_2 \wedge f_3\wedge f_4\,,\sum_{1\leq i<j\leq 4}\left(\frac{1}{|x_i-x_j|}\right)\, f_1\wedge f_2 \wedge f_5\wedge f_6}= 
\pscal{f_3\wedge f_4\,,\,\frac{1}{|x-y|}\,  f_5\wedge f_6}\,,
\end{equation}
\begin{multline}
\pscal{f_1\wedge f_2 \wedge f_3\wedge f_4\,,\left(\sum_{1\leq i<j\leq 4}\frac{1}{|x_i-x_j|}\right)\, f_1\wedge f_2 \wedge f_3\wedge f_6}\\
= \pscal{f_1\wedge f_4\,,\,\frac{1}{|x-y|}\,  f_1\wedge f_6}+\pscal{f_2\wedge f_4\,,\,\frac{1}{|x-y|}\,  f_2\wedge f_6}+\pscal{f_3\wedge f_4\,,\,\frac{1}{|x-y|}\, f_3\wedge f_6},
\end{multline}
and
\begin{multline}
\pscal{f_1\wedge f_2 \wedge f_3\wedge f_4\,,\,\left(\sum_{1\leq i<j\leq 4}\frac{1}{|x_i-x_j|}\right)\, f_1\wedge f_2 \wedge f_3\wedge f_4}\\
= \sum_{1\leq i<j\leq 4}\pscal{f_i\wedge f_j\,,\,\frac{1}{|x-y|}\, f_i\wedge f_j}\,.
\end{multline}
\end{lemma}

The proof of the above lemma  is based on long but straightforward computations.
We use the above result to compute some quantities that are needed  in order to perform the numerical computations of Section \ref{sec:num}.

We start with the computation of $\,\pscal{\tP_{sp}(R_0,R_1,R_2)\,,\, H \,\tP_{sp}(R_0,R_1,R_2)}$ which is needed to verify \eqref{condition_3D1}. By \eqref{eq:simplification_2_sp} we have 
$$ \pscal{\tP_{sp}(R_0,R_1,R_2), \, H \,\tP_{sp}(R_0,R_1,R_2)} = \big\langle \tP_{2,sp}(R_0,R_1,R_2)\,,\, H \,
\tP_{2,sp}(R_0,R_1,R_2)\big\rangle\,.$$
A simple calculation shows that
\begin{align}
&\big\langle \tP_{2,sp}(R_0,R_1,R_2)\,,\,\, H\,
\tP_{2,sp}(R_0,R_1,R_2)\big\rangle\nonumber\\
&\qquad=\int_{\R^3} \left( |\nabla s_0(R_0)(x)|^2+\frac{|\nabla s_0(R_1)(x)|^2+|\nabla p_1(R_2)(x)|^2}{2}\right)\,dx\label{eq:maria1}\\
&\qquad\qquad\qquad -\int_{\R^3} \frac4{|x|}\,\left( 2| s_0(R_0)(x)|^2+| s_0(R_1)(x)|^2+| p_1(R_2)(x)|^2\right)\,dx\label{eq:maria2}\\
&\qquad\qquad\qquad+\pscal{ \,^3\!P_{2,sp}(R_0,R_1,R_2)\,,\,\left(\sum_{1\leq i<j\leq 4}\frac{1}{|x_i-x_j|}\right)\,^3\!P_{2,sp}(R_0,R_1,R_2) }\,.\label{eq:maria3}
\end{align}
The well-known properties of the harmonic spherical functions show that
\begin{multline*}
\eqref{eq:maria1}+\eqref{eq:maria2}=\int_0^\infty \left( r^2\,R'_0(r)^2+ \frac{r^2R'_1(r)^2}{2} + \frac{r^4}{2}\left|\left(\frac{R_2(r)}{r}\right)'\right|^2\right) \,dr\\
-\int_0^\infty 4r\,\left( 2R_0(r)^2+ R_1(r)^2+R_2(r)^2\right) \,dr\,. 
\end{multline*}
On the other hand, using Lemma \ref{4to2}, the last integral \eqref{eq:maria3} is equal to
\begin{align*}
 \eqref{eq:maria3}&=
\iint_{\R^3\times \R^3} | s_0(R_0)(x)|^2 |s_0(R_0)(y)|^2+\big(2| s_0(R_0)(x)|^2+|p_1(R_2)(x)|^2\big)\,|s_0(R_1)(y)|^2\,dx\,dy\\
&\qquad+ \iint_{\R^3\times \R^3}2| s_0(R_0)(x)|^2 |p_1(R_2)(y)|^2\,dx\,dy\\
&\qquad-\iint_{\R^3\times \R^3}  \big( s_0(R_0)(x)s_0(R_0)(y) +  \overline{p_1(R_2)(x)}p_1(R_2)(y) \big)\,s_0(R_1)(x)s_0(R_1)(y)   \,dx\,dy\\
&\qquad-\iint_{\R^3\times \R^3}  s_0(R_0)(x)s_0(R_0)(y) \overline{p_1(R_2)(x)}p_1(R_2)(y)  \,dx\,dy\,.
\end{align*}
Using the well known formulae that can be found for instance in  Slater's book \cite{Slater-60} (Section 13-3 and Appendix 20a), this can be rewritten as
\begin{align*}
 \eqref{eq:maria3}&=\int_0^\infty\!\!\!\int_0^\infty\bigg(R_0(s)^2R_0(t)^2+\big(2R_0(s)^2+R_2(s)^2\big)\, R_1(t)^2  +2R_0(s)^2 R_2(t)^2 \bigg)\,\frac{s^2\,t^2\,ds\,dt}{\max\,\{s,t\}}\\
&\qquad-\int_0^\infty\!\!\!\int_0^\infty  R_0(s) R_0(t)R_1(s)R_1(t)\,\frac{s^2\,t^2\,ds\,dt}{\max\,\{s,t\}}\\
&\qquad-\int_0^\infty\!\!\!\int_0^\infty
 \big(    R_2(s)(R_2(t)\, \big( R_1(s)R_1(t)+R_0(s)R_0(t) \big)\,
\big)\,\frac{s^2\,t^2\,\min\,\{s,t\}\,ds\,dt}{3\,\max^2\,\{s,t\}}\,.
\end{align*}

\medskip

We then go on to calculate the expression of $\,\pscal{H\Psi,\tP_{sp}(R_0,R_1,\delta R)}$, appearing in \eqref{condition_3P1}. As a corollary of Lemma \ref{4to2}, of \eqref{3Psp}, and of \eqref{eq:simplification_2_sp}, we obtain
\begin{align*}
&\pscal{ \,^3\!P_{sp}(R_0,R_1,R_2)\,,\,\left(\sum_{1\leq i<j\leq 4}\frac{1}{|x_i-x_j|}\right)\,^3\!P_{sp}(R_0,R_1,\delta R) }\\ 
&\qquad\qquad= \pscal{ \,^3\!P_{2,sp}(R_0,R_1,R_2)\,,\,\left(\sum_{1\leq i<j\leq 4}\frac{1}{|x_i-x_j|}\right)\,^3\!P_{2,sp}(R_0,R_1,\delta R) }\\
&\qquad\qquad= \iint_{\R^3\times\R^3}\left(2|s_0(R_0)(x)|^2+|s_0(R_1)(x)|^2 \right)\,  p_1(R_2)(y)\, p_1(\delta R)(y)\,dx\,dy \\
&\qquad\qquad\qquad\qquad-\iint_{\R^3\times\R^3}  s_0(R_0)(x)\,p_1(\delta R)(x)\,  p_1(R_2)(y)\, s_0(R_0)(y)\,dx\,dy \\
&\qquad\qquad\qquad\qquad -\iint_{\R^3\times\R^3}  s_0(R_1)(x)\,p_1(\delta R)(x)\,  p_1(R_2)(y)\, s_0(R_1)(y)\,dx\,dy \,.
\end{align*}
Using again the formulae  in Slater's book \cite{Slater-60} (Section 13-3 and Appendix 20a), we finally deduce that
\begin{align*}
&\pscal{ \,^3\!P_{sp}(R_0,R_1,R_2)\,,\,\left(\sum_{1\leq i<j\leq 4}\frac{1}{|x_i-x_j|}\right)\,^3\!P_{sp}(R_0,R_1,\delta R) }\\ 
&\qquad\qquad=\int_0^\infty\int_0^\infty\left(2|R_0(s)|^2+|R_1(s)|^2 \right)\,  R_2(t)\, \delta R(t)\,\frac{s^2\,t^2}{\max{\{s,t\}}}\,ds\,dt\\
&\qquad\qquad\qquad\qquad -\int_0^\infty\int_0^\infty  R_0(s)\,\delta R(s)\,  R_2(t)\, R_0(t)\;\frac{s^2\,t^2\,\min{\{s,t\}}}{3\,\max^2{\{s,t\}}}\,ds\,dt\\
 &\qquad\qquad\qquad\qquad -\int_0^\infty\int_0^\infty  R_1(s)\,\delta R(s)\,  R_2(t)\, R_1(t)\;\frac{s^2\,t^2\,\min{\{s,t\}}}{3\,\max^2{\{s,t\}}}\,ds\,dt \,.
\end{align*}
On the other hand, by \eqref{eq:simplification_2_pd}, and using the same methods as above, we find
\begin{align*}
&\pscal{ \,^3\!P_{pd}(R_0,R_2,R_4)\,,\,\left(\sum_{1\leq i<j\leq 4}\frac{1}{|x_i-x_j|}\right)\,^3\!P_{sp}(R_0,R_1,\delta R) }\\ 
&\qquad\qquad=\pscal{ \,^3\!P_{2,pd}(R_0,R_2,R_4)\,,\,\left(\sum_{1\leq i<j\leq 4}\frac{1}{|x_i-x_j|}\right)\,^3\!P_{2,sp}(R_0,R_1,\delta R) }\\
&\qquad\qquad=-{\sqrt{2}}\, \int_0^\infty\int_0^\infty R_1(s)\,   R_2(s)\,\delta R(t)\,   R_4(t)\;\frac{s^2\,t^2\,\min{\{s,t\}}}{3\,\max^2{\{s,t\}}}\,ds\,dt \,.
\end{align*}

\bigskip\noindent{\bf Acknowledgments.} {\small M.J.E. and M.L. would like to thank Paul Indelicato who drew their attention to the problem studied in the present paper, and \'Eric S\'er\'e for interesting discussions. This work was partly done while the three authors were visiting the \emph{Institute for Mathematics and its Applications} of the {University of Minnesota} at Minneapolis (USA), during the annual program ``Mathematics and Chemistry''. M.J.E. and  M.L. have been supported by the ANR project \emph{ACCQuaRel} of the French Ministry of Research.}

%

\end{document}